%% file: main.tex
\documentclass{article}
\usepackage[utf8]{inputenc}

\usepackage{fullpage}

\usepackage{amsmath}
\usepackage{amsfonts}
\usepackage{amssymb}
\usepackage{amsthm}
\usepackage{bbm}

\usepackage[numbers]{natbib}

\usepackage{comment}

\usepackage{todonotes}
\usepackage{thmtools}

\DeclareMathOperator*{\argmax}{arg\,max}

\newcommand{\prob}{\texttt{Probed}}
\newcommand{\del}{\texttt{DEL}}
\newcommand{\opt}{\E[\texttt{OPT}]}
\newcommand{\E}{\mathbb{E}}
\newcommand{\Iin}{\mathcal I}

\newcommand{\supp}{\operatorname*{supp}}

\newtheorem{theorem}{Theorem}[section]
\newtheorem{lemma}[theorem]{Lemma}

\newtheorem{proposition}[theorem]{Proposition}

\theoremstyle{definition}
\newtheorem{definition}[theorem]{Definition}

\newtheorem{corollary}[theorem]{Corollary}

\title{Delegated Pandora's Box}
\author{
    Curtis Bechtel\thanks{supported by NSF Grants CCF-1350900 and CCF-2009060.}\\
    Department of Computer Science \\
    University of Southern California \\
    {\tt bechtel@usc.edu}
        \and
    Shaddin Dughmi\thanks{supported by NSF CAREER Award CCF-1350900 and NSF Grant CCF-2009060.} \\
    Department of Computer Science \\
    University of Southern California \\
    {\tt shaddin@usc.edu}
        \and
    Neel Patel\thanks{supported by NSF Grants CCF-1350900 and CCF-2009060.}\\
    Department of Computer Science \\
    University of Southern California \\
    {\tt neelbpat@usc.edu}
}
\date{\vspace{-5ex}}

\begin{document}

\maketitle


\input{abstract}

\input{introduction}

\input{preliminaries}

\input{model}

\input{standard-model}

\input{free-agent-discounted-cost-models}

\input{shared-costs-model}

\input{open-questions}

\bibliographystyle{abbrvnat}
\bibliography{references}

\newpage
\appendix

\input{lottery-mechanisms}

\input{missing-proofs}

\end{document}

%% file: abstract.tex
\begin{abstract}
    In delegation problems, a principal does not have the resources necessary to complete a particular task, so they delegate the task to an untrusted agent whose interests may differ from their own. Given any family of such problems and space of mechanisms for the principal to choose from, the delegation gap is the worst-case ratio of the principal's optimal utility when they delegate versus their optimal utility when solving the problem on their own. In this work, we consider the delegation gap of the generalized Pandora's box problem, a search problem in which searching for solutions incurs known costs and solutions are restricted by some downward-closed constraint. First, we show that there is a special case when all random variables have binary support for which there exist constant-factor delegation gaps for matroid constraints. However, there is no constant-factor delegation gap for even simple non-binary instances of the problem. Getting around this impossibility, we consider two variants: the free-agent model, in which the agent doesn't pay the cost of probing elements, and discounted-cost approximations, in which we discount all costs and aim for a bicriteria approximation of the discount factor and delegation gap. We show that there are constant-factor delegation gaps in the free-agent model with discounted-cost approximations for certain downward closed constraints and constant discount factors. However, constant delegation gaps can not be achieved under either variant alone. Finally, we consider another variant called the shared-cost model, in which the principal can choose how costs will be shared between them and the agent before delegating the search problem. We show that the shared-cost model exhibits a constant-factor delegation gap for certain downward closed constraints.
\end{abstract}

%% file: introduction.tex
\section{Introduction}
\label{introduction}

We take the natural next step in the study of delegated stochastic search problems involving multivariate decisions, constraints, and costs. The work of Bechtel and Dughmi \cite{bechtel2020delegated} has provided a fairly thorough understanding of principal-agent delegation in the presence of ``hard'' constraints on the search procedure --- scenarios of this form can be viewed as a principal \emph{delegating} a \emph{stochastic probing} problem to an agent. In this paper, we build a similar understanding when search is associated with cardinal costs instead. Scenarios of this form feature a principal who delegates, to an agent, a combinatorial generalization of the famous Pandora's box problem of Weitzman \cite{weitzman1979optimal}. As in the most relevant prior work on delegation, we imbue the principal with the power of commitment, rendering this a mechanism design problem.\footnote{In particular, a mechanism design problem without money.}

The conceptual starting point in this area is the work of Kleinberg and Kleinberg \cite{kleinberg2018delegated}, who consider a \emph{principal} delegating the selection of one option (we say \emph{element})  out of finitely many to an \emph{agent}. As a running example, consider a firm (the principal) delegating the selection of one job candidate out of many (the elements) to an outside recruitment agency (the agent). Each element is associated with a stochastic reward for both principal and agent, with independence across elements. The agent is tasked with ``exploring'' (we say \emph{probing}) these rewards and proposing one of them, which the principal may choose to accept or reject.

Problems of this form are most natural when exploration is not free, and \citet{kleinberg2018delegated} consider one model featuring a hard constraint on the number of options explored, and a second model featuring cardinal costs associated with exploration for both principal and agent. Bechtel and Dughmi~\cite{bechtel2020delegated} generalize the first model, in particular to settings in which exploration is combinatorially constrained (this is referred to as the \emph{outer} constraint), and multiple elements may be selected subject to another combinatorial constraint (this is referred to as the \emph{inner} constraint). When multiple elements are selected, rewards are additive for both the principal and the agent. In this paper, we similarly generalize the second model of \cite{kleinberg2018delegated}: there is no outer constraint on exploration, but rather per-element probing costs for the principal and agent. Moreover, there again is an inner constraint (which we will often refer to simply as the constraint) on the set of elements selected. Rewards and probing costs are now both additive across elements. The problem being delegated here is a generalized Pandora's box problem, as in \cite{singla2018price}.

There are multiple natural ways of instantiating the utilities of both the principal and the agent, depending on who we assume incurs the exploration (i.e., probing) costs. Some ways in which costs may be shared include:
\begin{itemize}
\item The principal and agent each pay a fixed percentage of the total probing cost. In our running example, the recruitment agency may have a policy in which they only pay a fixed fraction of the cost of interviewing each candidate. Such scenarios fall under our first model which we refer to as the \emph{standard model} of utilities\footnote{As long as the principal or the agent does not pay the entire cost in the standard model, we can re-scale their utilities and without loss of generality, assume that they both pay the equal cost.}. \citet{kleinberg2018delegated} assumes cost-sharing according to the standard model of utilities.

\item The principal pays the full cost of exploration. In our running example, recruitment agencies only commit to investing their time and expertise, where the principal bears the entire cost of exploration. We refer this model as the \emph{free-agent} model.

\item The principal chooses as part of their strategy how individual costs are shared. In our running example, the principal may be willing to pay a large fraction of the cost of interviewing good (in expectation) candidates, but still allows the agent to interview bad (in expectation) candidates so long as they bear most of the cost. We refer to this model as the \emph{shared-cost} model. This model provides the principal with much more power when delegating.
\end{itemize}
No matter our utility model and cost model, we seek mechanisms which approximate the principal's optimal \emph{non-delegated utility}: the maximum expected utility the principal can obtain by solving the search problem themselves. When such a mechanism matches the non-delegated utility up to a factor $\alpha$, we refer to it as an \emph{$\alpha$-factor} mechanism.

\subsection*{Our Models and Results}

In our first model --- which we refer to as the \emph{standard model} of utilities --- we follow in the footsteps of \cite{kleinberg2018delegated} by incorporating the exploration costs into both the principal and agent's utilities.\footnote{Whereas it is not uncommon for the agent in delegation to bear the costs of completing the task, this model also incorporates the costs into the principal's objective. This can capture a principal concerned with optimizing a social objective, as well as scenarios in which probing costs are shared equally by the principal and agent.}
Our results for this model are a mixed bag: When each element's reward distribution has binary support, we obtain constant-approximate delegation mechanisms when the constraint is a matroid. The proof proceeds via a reduction to the matroid prophet inequalities against an almighty adversary from \cite{feldman2016online}. This result generalizes the result of \cite{kleinberg2018delegated} for their second model, which also features binary distributions.
On the other hand, we obtain strong impossibility results for non-binary distributions, ruling out any sublinear (in the number of elements) approximation to the principal's optimal non-delegated utility, even for the rank one matroid. This shows that the result of \cite{kleinberg2018delegated} for their second model, which also features a rank one matroid constraint, can not be generalized to non-binary distributions. Even more emphatically, we rule out certain bicriteria approximations for the standard model of utilities: even if  probing costs are discounted by any absolute constant, the principal's delegated utility can not approximate --- up to any constant --- their undelegated utility in the undiscounted setting.

Motivated by our impossibility results for the standard model of utilities, we explore models in which exploration costs are shared unequally between the principal and the agent. In the \emph{free-agent model}, the agent incurs no exploration costs, which are born entirely by the principal. For various constraints such as matroids, matchings, and knapsacks, we obtain bicriteria approximate mechanisms of the following form for various pairs of constants $\alpha, \delta$: the principal's delegated utility in the setting where probing costs are discounted by $\delta$ matches, up to a factor of $\alpha$, their optimal undelegated utility in the undiscounted setting. Our results proceed by reduction to the online contention resolution schemes against an almighty adversary from \cite{feldman2016online}. We complement this with a negative result, ruling out the traditional uni-criteria constant-approximate mechanisms. Specifically, absent any discount on probing costs, no delegation mechanism approximates  the principal's optimal undelegated utility up to any constant.

Our final utility model allows the principal to declare, up front as part of their mechanism, an arbitrary split of the probing cost for each element between the principal and the agent. We refer to this as the \emph{shared-cost model}. This turns out to be the most permissive of our models: for constraints including matroids, matchings, and knapsacks, we obtain delegation mechanisms which approximately match, up to a constant, the principal's optimal undelegated utility.\footnote{We note that, since the principal can offload much of the costs of exploration to the agent, there exist instances in which the principal's delegated utility strictly exceeds their undelegated utility. However, we also show that there are simple instances in which the principal's delegated utility is necessarily less than their undelegated utility, ruling out general results with approximation factors exceeding $1$.} Our results here are again by reduction to online contention resolution schemes against an almighty adversary from \cite{feldman2016online}.

Lastly, we also begin a preliminary exploration of randomized mechanisms for delegating the generalized Pandora's box problem. We obtain negative results for a restricted class of randomized mechanisms, and leave open the general question of whether randomization yields significantly more power in this setting, including whether it overcomes some of our impossibility results for deterministic mechanisms.

\subsection*{Additional Discussion of Related Work}
For additional discussion of related work pertaining to delegation, stochastic probing problems, prophet inequalities, and contention resolution, we refer the reader to \cite{bechtel2020delegated}. Also relevant to this paper is the work on generalizations of the Pandora's box problem. In particular, Singla \cite{singla2018price} introduces a model generalizing the $1$-uniform matroid ``inner'' constraint to arbitrary downward-closed constraints and proposes constant-factor algorithms for matroids, matchings, and knapsack constraints. Gamlath et. al. \cite{gamlath2019beating} further improves approximation guarantees for the generalized Pandora's box problem with matching constraints.

%% file: preliminaries.tex
\section{Preliminaries}
\label{preliminaries}

\subsection{Pandora's Box}

Weitzman’s Pandora’s box problem \cite{weitzman1979optimal} is defined as follows: given probability distributions of $n$ independent random variables $X_1, \dots, X_n$ over $\mathbb R_{\ge 0}$ and their respective probing costs $c_1, \dots, c_n$, adaptively probe a subset $\prob \subseteq [n]$ that maximizes the expected utility:
\begin{equation}
	\E \left[ \max_{i \in \prob} \{X_i\} - \sum_{i\in \prob} c_i \right].
\end{equation}

Weitzman \cite{weitzman1979optimal} proposes a simple but optimal strategy for maximizing expected utility. For each element $i \in [n]$, this strategy chooses a \emph{cap value} (sometimes called priority value or surplus value) $\tau_i$ satisfying $\E[(X_i - \tau_i)^+] = c_i$. Then it probes elements in decreasing order of cap value, stopping the first time that the largest observed $X_i$ value exceeds the largest unprobed cap value. Finally, it selects the element $i$ with maximum observed $X_i$.

In this work, we focus on the more general version of the Pandora's box problem defined in \cite{singla2018price}. We are given a set of elements $E$ and a downward-closed constraint $\Iin \subseteq 2^E$ over the ground set $E$. The goal is to adaptively probe a set of elements $\prob$ and select a set of feasible elements $S \subseteq \prob$ for which $S \in \Iin$ that maximizes the following objective:
\begin{equation}
	\E \left[ \sum_{i \in S} X_i - \sum_{i \in \prob} c_i\right]. \label{eq:objective}
\end{equation}
For the remainder of the paper, we will write $X(S) = \sum_{i \in S} X_i$ and $c(S) = \sum_{i \in S} c_i$ in any setting with utilities $\{X_i\}_{i \in E}$ and costs $\{c_i\}_{i \in E}$. We will also refer to $(i, x)$ for any $i \in E$ and $x \in \mathbb R_{\ge 0}$ as a possible \emph{outcome} or \emph{realization} of element $i$.

Singla \cite{singla2018price} proposes constant-factor approximation algorithms for the general Pandora's box problem for many constraints. In particular, these algorithms are optimal for matroids and $2$-approximate for both matching and knapsack constraints. 


\subsection{Greedy Prophet Inequality}
An instance of the generalized prophet inequality problem is given by a set system $\mathcal M$ with ground set $E$ and feasible sets $\mathcal I$ and independent random variables $X_i$ supported on $\mathbb R_{\ge 0}$ for all $i \in E$. We take the perspective of the \emph{gambler}, who knows $\mathcal M$ and the distributions of the random variables $\{X_i\}_{i \in E}$. The gambler starts with an empty set $S$ of accepted elements and then observes each element in $E$ in an order chosen by an adversary. For the purposes of this paper, we play against the \emph{almighty adversary} defined in \cite{feldman2016online}, the strongest possible adversary, who knows all the coin flips of the gambler's strategy. When the element $i \in E$ arrives, the gambler learns the realization of $X_i$ and has to decide online whether to accept element $i$ or not based on $(i,x_i)$ and the previously accepted elements $S$. However, they can only accept $i$ if $S \cup \{i\}$ is feasible in $\mathcal M$. The gambler seeks to maximize their utility $\E (X(S)) = \E \left[ \sum_{i \in S} X_i \right]$, and in particular to compete with a \emph{prophet} who plays the same game and knows the realizations of all random variables in advance. If the gambler has a strategy guaranteeing an $\alpha$ fraction of the prophet’s expected utility in expectation, we say that we have an $\alpha$-factor prophet inequality.
We now define a particular class of strategies for the gambler:
\begin{definition}[\emph{Greedy monotone strategy} $\mathcal A_t$]
A greedy monotone strategy $\mathcal A_t$ for the gambler is described by choice of thresholds $t = \{t_i : i\in E\}$ and a downward closed system $\mathcal I_t\subseteq \mathcal I$, and can be expressed as $\mathcal A_t = \{\{(i, x_i) : i \in S \} : S \in \mathcal I_{t} \text{~and~} x_i \geq t_i \text{~for all~} i \in S\}$. A gambler following $\mathcal A_t$ accepts element $i$ with outcome $(i, x_i)$ if and only if $x_i \geq t_i$ and set of elements accepted so far along with the element $i$ stays in $\mathcal I_t$.
\end{definition}
Greedy monotone strategies for the gambler is proposed in \cite{feldman2016online} for matroid, matching, and knapsack constraints that achieve $1/4$, $1/2e$, and $3/2-\sqrt 2$ factor prophet inequality respectively.


\subsection{c-Selectable Greedy OCRS Schemes}
 We will give a brief overview of online contention resolution schemes \cite{feldman2016online} in this section.Gi ven a downward-closed family $\mathcal I$ over the ground set of elements $E$ with $|E| = n$, let $P_{\mathcal I} \subseteq [0, 1]^n$ be the convex hull of the indicator vectors of all feasible sets: $P_{\mathcal I} = \operatorname{conv}(\{\mathrm{1}_F : F \in \mathcal I\})$. We say that a convex polytope $P \subseteq [0,1]^n$ is a \emph{relaxation} of $P_\mathcal I$ if it contains the same $\{0,1\}$-points, i.e. $P \cap \{0,1\}^n = P_{\mathcal I} \cap \{0,1\}^n$.

Consider the following online problem: given some $\mathcal I$ as above and some $x \in P_{\mathcal I}$, let $R(x)$ be a random subset of \emph{active} elements, where each element $i \in E$ is active with probability $x_i$ independently of all others. The elements in $E$ are revealed online in an order chosen by an adversary, and when each element $i$ is revealed, we learn whether or not $i \in R(x)$. After we learn the state of element $i$, we must irrevocably decide whether or not to select $i$. An OCRS for $P$ is an online algorithm that selects a subset $S \subseteq R(x)$ such that $S \in \mathcal I$.

\begin{definition}[Greedy $c$-selectable OCRS]
	Let $P \subseteq [0,1]^n$ be a relaxation of $P_{\mathcal I}$. A greedy OCRS $\pi$ for $P$ is an OCRS that for any $x \in P$ defines a downward-closed family of sets ${\mathcal I}_x \subseteq \mathcal I$. Then an active element $i$ is selected if, together with the already selected elements, the obtained set is in $\mathcal I_x$. Moreover, we say the greedy OCRS is $c$-selectable if for all $x \in P$ and $i \in E$
	\begin{equation*}
		\Pr[I \cup \{i\} \in \mathcal I_x \text{~for all~} I \subseteq R(x) \text{~and~} I \in \mathcal I_x] \geq c.
	\end{equation*}
\end{definition}

%% file: model.tex
\section{Delegation Model}
\label{sec:model}

In this paper, we will use several slightly different models of delegation which can be viewed as variants of a single \emph{standard model} of delegated Pandora's box. This model formally consists of: two players called the \emph{principal} and the \emph{agent}; a ground set of \emph{elements} $E$; for each element $i \in E$, an independent distribution $\mu_i$ over $\mathbb{R}_{\ge 0} \times \mathbb{R}_{\ge 0}$ giving possible utility pairs for the principal and agent, respectively; for each element $i \in E$, a \emph{probing cost} $c_i \in \mathbb{R}_{\ge 0}$; and a downward-closed set system $\mathcal{M} = (E, \mathcal{I})$ with feasible sets $\mathcal{I}$ over the ground set $E$ (i.e. $\mathcal I \subseteq 2^E$ and if $S \in \mathcal I$ then $T \in \mathcal I$ for any $T \subseteq S$). 

Given an element $i$, we let $X_i$ and $Y_i$ be random variables denoting the random value obtained for the principal and agent from element $i$ with joint distribution $(X_i, Y_i) \sim \mu_i$, where $X_i$ and $Y_i$ may be arbitrarily correlated but are independent of random variables from other elements. For any $(x, y) \in \supp(\mu_i)$, we call $(i, x, y)$ an \emph{outcome} or \emph{realization} of element $i$. For any set of outcomes $\mathcal S = \{ (i_1, x_1, y_1), \dots, (i_k, x_k, y_k) \}$ such that $S = \{ i_1, \dots, i_k \} \in \mathcal{I}$ for distinct $i_1, \dots, i_k$, we call $\mathcal S$ a \emph{solution}. In general, we denote the set of all possible outcomes as $\Omega = \{ (i, x, y) : (x, y) \in \supp(\mu_i), i \in E \}$ and the set of all solutions with respect to the constraint $\mathcal I$ as $\Omega_{\mathcal I} \subseteq 2^\Omega$.

Given such an instance as described above, the principal and agent play an asymmetric game in which the principal alone has the power to choose the mechanism and accept a solution, and the agent alone has the power to search for solutions. More specifically, in order to learn about the true realization $(X_i, Y_i)$ of an element $i$, the agent can \emph{probe} element $i$. We allow them to probe elements adaptively, choosing what to probe next based on previously realized outcomes. Let us say that the agent ultimately probes the set $\prob \subseteq E$, obtaining outcomes $T$. Depending on the mechanism, they can choose to share information about $T$ with the principal. The principal can \emph{accept} any valid solution $S \subseteq T$, yielding a net utility of $\sum_{(i, x, y) \in S} x - \sum_{i \in \prob} c_i$ for the principal and $\sum_{(i, x, y) \in S} y - \sum_{i \in \prob} c_i$ for the agent. The principal can alternatively choose to \emph{reject} all solutions and maintain the status quo, yielding a net utility of $- \sum_{i \in \prob} c_i$ for both players. Both players have common knowledge of the setup of the problem, including all distributions $\{ \mu_i \}_{i \in E}$ but excluding the true realizations of elements, and they each act to maximize their own expected utility.

As in the models from previous work, we assume that the agent cannot lie by misrepresenting the utilities of a probed outcome or by claiming to have probed an unprobed element. We believe that this is a natural assumption in many settings where outcomes can be easily verified by the principal. Additionally, we assume that the principal has commitment power, i.e. the agent can trust the principal to follow the rules of whatever mechanism they choose. The principal can force the agent to also follow the rules of the mechanism insofar as they can detect violations of the rules. Finally, we also assume that all instances of this problem satisfy $\E[X_i] > c_i$ and $\E[Y_i] > c_i$ for all $i\in E$. The first assumption is without loss of generality, since $\E[X_i] \le c_i$ would imply that the principal has no incentive to probe or accept element $i$, so the agent would not probe it either. The second assumption allows us to avoid uninteresting impossibilities for the delegation gap defined in Section \ref{delegation-gap}, since $\E[Y_i] \le c_i$ would imply that the agent has no incentive to probe or propose element $i$ but the principal may still be able to receive a lot of utility from element $i$.

For this paper, we're interested in \emph{single-proposal} mechanisms as defined in \cite{kleinberg2018delegated} and used in \cite{bechtel2020delegated}. A single proposal mechanism consists of an \emph{acceptable set} $\mathcal{R} \subseteq \Omega_{\mathcal I}$ containing all solutions that the principal is willing to accept. In such a mechanism, the principal starts by declaring their choice of $\mathcal{R}$. The agent responds by adaptively probing any set of elements $\prob \subseteq E$ of their choosing, receiving the set of outcomes $T = \{ (i, X_i, Y_i) : i \in \prob \}$. Once they are done probing, they can \emph{propose} some valid solution $\mathcal S \subseteq T$ to the principal. Finally, the principal can either accept or reject the solution $\mathcal S$. If $\mathcal S$ is not a valid solution, $\mathcal S$ contains misrepresentations of the truth, or $\mathcal S \notin \mathcal{R}$, then the principal must reject $\mathcal S$. We note that this mechanism is deterministic in the sense that the principal chooses a deterministic $\mathcal{R}$ and their response to the agent's choices is deterministic. This is in contrast to the randomized mechanisms discussed briefly after Theorem \ref{thm:efficient_delegation_from_OCRS} and lottery mechanisms as defined in Appendix \ref{lottery-mechanisms}.

Given element $i$, we define the \emph{cap value} or \emph{surplus value} for the principal $\tau^x_i$ as the solution to $\E[(X_i - \tau^x_i)_+] = c_i$. We further define \emph{truncated random variables} $Z^{\min}_i = \min \{ X_i, \tau^x_i \}$ for the principal for all $i\in E$.
We similarly define agent's cap values $\tau^y_i$ as the solution to $\E[(Y_i - \tau^y_i)_+] = c_i$, and the truncated random variable for the agent as $W^{\min}_i = \min \{ Y_i, \tau^y_i \}$. Note that the expected utility (including the probing cost) of a particular element is negative for the elements with negative cap values. We sometimes drop superscript from the principal's cap values and denote $\tau_i^x$ as $\tau_i$ for $i\in E$ whenever it is clear.

\subsection{Delegation Gap}
\label{delegation-gap}

As in \cite{bechtel2020delegated,kleinberg2018delegated}, we are not interested in finding optimal delegation mechanisms so much as finding delegation mechanisms that approximate the principal's optimal \emph{non-delegated} utility. The optimal non-delegated utility refers to the principal's optimal utility when delegating to an agent who shares their interests (alternatively, their optimal utility when they act as both the principal and agent, i.e. they have the power to probe elements and accept solutions). Note that the non-delegated problem that the principal faces is exactly the generalized Pandora's box problem with a downward-closed constraint. Therefore, our main model is a delegated version of this problem, hence why we call it the delegated Pandora's box problem.

Let $\opt$ be the principal's optimal non-delegated utility. Singla \cite{singla2018price} shows that for any downward closed constraint $\mathcal I$,
\begin{equation*}
    \opt \leq \E\left[\max_{S\in \mathcal I}\sum_{i\in S} Z_i^{\min}\right].
\end{equation*}
Let $\E[\del_{\mathcal R}]$ be the expected utility of the delegating principal with single-proposal mechanism $\mathcal R$, i.e. the expected utility of the principal who delegates with acceptable set $\mathcal R$ to an agent who acts in order to maximize their own expected utility given $\mathcal R$.

Now, we define $\alpha$-factor delegation strategies, which guarantee the principal at least an $\alpha$-factor of $\opt$ when they delegate.
\begin{definition}
	Fix an instance of the delegated Pandora's box problem. We say that a mechanism $\mathcal{R}$ is an $\alpha$-\emph{factor delegation strategy} for $\alpha \in [0, 1]$ if
	\begin{equation*}
		\E[\del_{\mathcal R}] \geq \alpha \cdot \opt.
	\end{equation*}
	Moreover, we say $\mathcal R$ is an $\alpha$-factor \emph{agent-agnostic} strategy if $\E[\del_{\mathcal R} ]\geq \alpha \cdot \opt$ for all instances with the same costs and marginal distributions of the principal's values $\{X_i\}_{i \in E}$, regardless of the distribution of the agent's values $\{Y_i\}_{i \in E}$.
\end{definition}
We sometimes refer to $\alpha$-factor strategies as $\alpha$-delegation and $\alpha$-factor agent-agnostic strategies as $\alpha$ agent-agnostic delegation. Note that if $\alpha$-factor agent-agnostic strategies exist for the principal, then the principal can obtain an $\alpha$-factor of $\opt$ even when they do not have any information about the distribution of $\{Y_i\}_{i \in E}$.

Now, we define the delegation gap of the family of instances of delegated Pandora's box.
\begin{definition}
	The \emph{delegation gap} of a family of instances of delegated Pandora's box is the minimum, over all instances in the family, of the maximum $\alpha$ such that there exists an $\alpha$-factor strategy for that instance. This gap measures the minimum fraction of the principal's non-delegated utility they can achieve when delegating optimally. We similarly define the agent-agnostic delegation gap for agent-agnostic delegation.
\end{definition}

\subsection{More General Mechanisms}
\label{other-mechanisms}

Having now defined our model and the space of single-proposal mechanisms, it is natural to ask about the power and generality of such mechanisms. It might be beneficial for the principal to consider a larger class of mechanisms that have, for example, more signals to choose from and multiple rounds of communication. However, as in previous work on delegation and similar mechanism-design problems, we argue that that any \emph{multi-round signaling mechanism} can be equivalently implemented by a single-proposal mechanism. This allows us to consider only single-proposal mechanisms without loss of generality. Since this type of argument is similar to the revelation principle and is very common in the literature \cite{alonso2008optimal, armstrong2010model, bechtel2020delegated, kleinberg2018delegated}, we will include only an informal sketch here.

Consider any multi-round signaling mechanism $M$. We will construct a single-proposal mechanism $S$ that simulates $M$. In $S$, the principal commits to accepting any solution that they could accept when both players follow $M$. Since the agent following $M$ can predict this set of acceptable solutions and the sequence of probes and signals leading to any such solution, they can act in a way that optimizes their expected utility given the solutions that the principal would accept. Therefore, the agent responding to $S$ can do no better than following the same such optimal sequence of probes and then proposing whichever solution the principal would have accepted under $M$. Since they can do just as well under $S$ and have no reason to deviate from the optimal strategy of $M$, these mechanisms are equivalent.

We note here that this argument applies to deterministic mechanisms. Lottery mechanisms as defined in Appendix \ref{lottery-mechanisms} could have strictly more power than their deterministic counterparts.

\subsection{Model Variants} \label{sec:model_variants}
\label{model-variants}

In this paper, we consider a few different variants of the model and approximation measure as defined above. The first such variant, called the \emph{binary model}, is just a special case of delegated Pandora's box in which the distribution $\mu_i$ of every element $i$ has support for exactly two outcomes: $\bot = (i, 0, 0)$ and $\omega_i = (i, x_i, y_i)$. A simpler version of this model in which the inner constraint is a $1$-uniform matroid was investigated in \cite{kleinberg2018delegated}, and we extend their definition to general matroid inner constraints. As motivation for this model, we consider search problems in which the principal and agent know the full space of possible outcomes but don't know which of those outcomes are feasible. However, they both share a prior probability on the feasibility of each outcome, all outcomes are mutually independent, and the agent can check the feasibility of any element by paying a probing cost. This is also an extension of prior work as described in the introduction.

Second, we consider the \emph{free-agent model}. This model changes only the utility of the agent such that they do not pay the cost of any probed elements. In order to ensure that the agent does not probe all elements and incur too a large cost for the principal, we assume that the agent breaks ties in favor of the principal when deciding what element to probe next. Therefore, if the principal doesn't accept any outcomes from a particular element, then they know that the agent will not probe that element. We motivate this model both by negative results in the standard model and by settings in which the principal is constrained in advance to cover all costs that the agent may incur, e.g. an employer that commits to reimbursing employees for all work-related costs.

Third, we consider \emph{discounted-cost approximations}, a new measure of approximation for delegated Pandora's box problems. Given an instance $I$ of any model of delegated Pandora's box and some \emph{discount factor} $\delta$, consider a new instance $J$ identical to $I$ except that the cost of each element $i$ is $(1 - \delta) c_i$, where $c_i$ is the original cost. 
\begin{definition}
    We say that a mechanism $\mathcal R$ is an $(\alpha,\delta)$-factor delegation strategy if the principal's delegated utility in the $\delta$-discounted instance $J$ is at least an $\alpha$-factor of their non-delegated utility in the original instance $I$.
\end{definition}
Observe that this is a bi-criteria approximation in which we aim to minimize $\delta$ and maximize $\alpha$. This approximation measure can be used as a means of determining how far the principal's costs are from being able to achieve a constant delegation gap. We additionally motivate it by settings in which the agent pays a smaller cost for searching than the principal would, e.g. a contractor which, through prior experience or economies of scale, is able to save on costs and share these savings with the contractee.

Finally, we consider the \emph{shared-cost model}. This model considers a fixed cost to probe each element that the principal can pay alone or share with the agent. In particular, it allows the principal to set the agent's cost $c'_i$ for element $i$. These costs are announced to the agent along with the acceptable set $\mathcal{R}$. Then, if the agent probes element $i$, they pay a cost of $c'_i$ and the principal pays the remaining cost for that element, i.e. $c_i - c'_i$. To avoid direct transfers of value between the principal and agent, the principal can only choose $0 \le c'_i \le c_i$ so that both costs are nonnegative. We briefly observe that there are instances of this model for which the principal's optimal delegated utility is strictly greater than their optimal non-delegated utility. This is easy to see by considering any instance for which $X_i = Y_i$ for all elements $i$: the principal can set $c'_i = c_i$ and have the agent run their optimal non-delegated strategy while they do not pay any of the costs. Therefore, the delegation gap $\alpha$ of such instances can be greater than $1$. We introduce this model in the hopes that the principal's increased power can lead to better approximations. Furthermore, this model resembles settings in which the principal can choose different reimbursement amounts for each of the agent's actions, but is unable to reimburse more than the true cost (no direct transfers).

%% file: standard-model.tex
\section{Standard Model Delegation}
\label{standard-model}

 In this section, we consider the delegation gap of the standard model of the delegated Pandora's box problem. We start by looking at the binary model special case, and show that this model has constant-factor delegation gaps for matroid constraints. Then, in Section \ref{standard-model-impossibility}, we show that the standard model (without binary assumption on $\mu_i$) does not admit constant delegation gaps in general, even for the rank one matroid. Before getting to the main result for this model, we analyze the (non-delegated) Pandora's box problem with exogenous order as discussed in \cite{kleinberg2018delegated} for rank one matroids, and extend their result to more general constraints.

\subsection{Non-delegated Generalized Pandora's Box with Exogenous Sequence}
\label{binary-model-preliminaries}

Consider a variant of the generalized Pandora's box problem, which we will call \emph{generalized Pandora's box with exogenous order}, in which the searcher is limited to consider elements in an order that is specified in advance a part of the instance. For each element in this order, the searcher can choose to skip the element without probing, or probe the element and either accept or reject based on the realization. Once the searcher makes a decision about the current element, they cannot undo this decision. This is an extension of a similarly-named model from \cite{kleinberg2018delegated}. We now define the threshold strategy for Pandora's box problem with exogenous ordering. Recall that the cap value $\tau_i$ for an element $i$ is defined by $\E[(X_i - \tau_i)^+] = c_i$, where $X_i$ is the random value of the element and $c_i$ is its cost.

\begin{definition}[Threshold Strategy $(\mathcal A, \{\tau_i\} , \{X_i\})$] \label{def:threshold_picking}
	Given a downward-closed family of solutions $\mathcal A$, the \emph{threshold strategy} defined by $\mathcal A$ functions as follows: Consider the searcher who has already accepted outcomes $\mathcal S = \{ (i_1, x_1), \dots (i_k, x_k) \}$ and is deciding what to do about element $i$. They should probe element $i$ if and only if $\mathcal S \cup (i, \tau_i) \in \mathcal A$. Furthermore, they should accept element $i$ if and only if $\mathcal S \cup (i, X_i) \in \mathcal A$.
\end{definition}

With this type of strategy in mind, we can extend the approximation of this problem from rank one matroids in \cite{kleinberg2018delegated} to more general downwards closed constraints. Lemma \ref{lem:reduction-pandora-to-prophet}, which is a corollary of \cite[Theorem 5]{esfandiari2019online}, provides a reduction from generalized Pandora's box with exogenous ordering for arbitrary downwards closed constraints to adversarial greedy prophet inequalities.

\begin{lemma}\label{lem:reduction-pandora-to-prophet}
	Let $J$ be an instance of the generalized prophet inequality problem with random variable $Z_i^{\min}=\min\{X_i,\tau_i\}$ for all $i \in E$ and constraint $\mathcal I$. If there exists an $\alpha$-factor greedy monotone strategy for $J$ against the almighty adversary, then there exists an $\alpha$-factor threshold strategy for the Pandora's box instance $I = (E, \{X_i\}, \Iin, \{c_i\})$ with exogenous ordering.
\end{lemma}
\begin{proof}
    Corollary of Theorem 5 from \cite{esfandiari2019online}.

\end{proof}


\subsection{Binary Model: Efficient Delegation for Matroids}
\label{binary-model}
Singla \cite{singla2018price} proposes an optimal strategy for Pandora's box with a matroid constraint that can be simplified in the binary setting as follows: probe elements one by one starting from the element with the maximum cap value. Given currently selected elements $S$, probe the next element with the maximum cap value $i$ such that $S\cup i \in \mathcal I$. After probing the element $i$, select $i$ if and only if $X_i>0$.

Consider the binary delegated Pandora's box instance for constraint $\mathcal I$ where the distributions $\mu_i$ of every element $i \in E$ has support on exactly two outcomes: $\bot = (i, 0, 0)$ and $\omega_i = (i, x_i, y_i)$. In the following Theorem, we show that the principal can design a $1/4$-factor strategy $\mathcal R$ for the standard delegation model for $\mu_i$ with binary support and a matroid constraint. The key idea is to use the reduction from Pandora's box with an exogenous order to prophet inequalities as described in Lemma~\ref{lem:reduction-pandora-to-prophet}.
\begin{theorem}
  	There exists a $1/4$-factor strategy for the binary model of delegated Pandora's box with a matroid constraint.
\end{theorem}
\begin{proof}
	Take an instance of the binary model with elements $E$ such that for all $i \in E$, we have $(X_i,Y_i)=(x_i,y_i)$ with probability $p_i$ and $(X_i, Y_i) = (0, 0)$ otherwise. Consider a $1/4$-approximate greedy monotone strategy, as proposed in \cite{feldman2016online}, for the prophet inequality instance with random variables $Z_i^{\min} = \min\{X_i,\tau^x_i\}$ for all $i \in E$ and matroid constraint $\mathcal I$ against the almighty adversary. This strategy is defined by thresholds $t = \{t_i\}_{i \in E}$ and a matroid constraint $\mathcal I_{t}\subseteq \mathcal I$. Given any order of arrival of elements, the gambler selects element $i$ if and only if $Z^{\min}_i \geq t_i$ and the set of all accepted elements (including element $i$) is contained in $\mathcal I_{t}$. Without loss of generality, we assume that $t_i$ is such that $0 < t_i \le x_i$ for all $i \in E$. This is because the gambler has no incentive to accept elements of value $0$ and $\tau^x_i < x_i$ due to the assumption $\E[X_i]>c_i$.
	
	Given thresholds $\{ t_i \}_{i \in E}$, the principal restricts the agent to elements in the set $E'= \{i\in E: \tau^x_i \geq t_i\}$. Let $\mathcal I_{t}^{E'}$ be the matroid constraint obtained by restricting $\mathcal I_{t}$ to the set of elements $E'\subseteq E$. We can describe the gambler's greedy monotone strategy as $\mathcal A = \{\{(i,z_i):i \in S \land z_i\geq t_i \}: S\in \mathcal I^{E'}_t\}$. Now, we define the principal's single proposal mechanism as follows:
	\begin{equation*}
	    \mathcal R = \{\{(i,x_i,y_i): i\in S \}: S\in \mathcal I^{E'}_{t} \text{ and } x_i \geq t_i ~\forall i\in S\}.
	\end{equation*}
	
 	For all $i\in E'$, $\mu_i$ has binary support, so $Y_i \geq \tau^y_i$ implies that $X_i \geq t_i$, where $\tau_i^y$ is the agent's cap value for element $i$ satisfying $\E[(Y_i - \tau^y_i)_+]=c_i$. Given this set of acceptable solutions $\mathcal R$, the agent faces an instance of Pandora's box on the set of elements $E'$ with matroid constraint $\mathcal I^{E'}_{t}$. Therefore, the agent's optimal strategy can be described as follows \cite{singla2018price}: given the current set of accepted elements $S\subseteq E'$ with $S\in \mathcal I^{E'}_{t}$, probe an element $i \in E' \setminus S$ such that $S \cup i \in \mathcal I^{E'}_{t}$ and $\tau^y_i$ is maximal. Then they will accept element $i$ if and only if $Y_i \geq \tau_i^y$, which is equivalent to selecting element $i$ if and only if $X_i \geq t_i$. Thus, the agent simply implements the threshold strategy $(\mathcal A,\{\tau_i\},\{X_i\})$ for the principal's Pandora's box instance with exogenous order equal to their probing order. Therefore, by Lemma~\ref{lem:reduction-pandora-to-prophet}, we conclude that the principal's expected delegated utility $\E[\del_{\mathcal R}] \geq 1/4 \cdot \opt$.
\end{proof}

\subsection{Standard Model Impossibility}
\label{standard-model-impossibility}

Now we will consider the standard model of delegated Pandora's box and show that this problem does not have constant-factor delegation gaps in general, even for rank one matroid constraints. In Proposition \ref{prop:impos1}, we present a family of instances of delegated Pandora's box for which the delegation gap is $O(1/n)$ where $n$ is the number of elements. The main challenge in this model is when the agent pays to probe, the principal needs to construct their acceptable set $\mathcal R$ such that the agent has an incentive to probe all desirable elements. For example, consider an element $i$ for which $c_i = 1 / \sqrt n$, $X_i = n$ with probability $1 / n$ and otherwise $X_i = 0$, and $Y_i = n$ independently with probability $1 / n$ and otherwise $Y_i = 0$. In this case, if the principal only accepts the outcome $X_i = n$, then the agent will not probe element $i$ because their expected utility from probing is $n \times \Pr[X_i=n] \Pr[Y_i=n] - 1 / \sqrt n < 0$ for $n > 1$. In order to ensure that the agent probes such elements, the principal might have to accept undesirable outcomes where $X_i = 0$. Hence, if there are multiple such elements then the principal ends up accepting unwanted outcomes with a high probability that leads to $O(1/n)$ delegation gap. The following Proposition shows the claim formally.

\begin{proposition}\label{prop:impos1}
	There exist instances of the standard model of delegated Pandora's box on $n$ elements for which the delegation gap is $O(\frac{1}{n})$.
\end{proposition}
\begin{proof}
	For any positive integer $n > 1$ and real $0 < \varepsilon \le \frac{1}{2n}$, let $M$ be a positive integer such that $M \ge n / \varepsilon$ and consider the following instance of delegated Pandora's box. We have $n$ identical elements $E = \{1, \dots, n\}$ where each element $i$ has a probing cost $c_i = 1 - \varepsilon$ and random utilities $(X_i, Y_i) \sim \mu_i$. The principal's utility $X_i$ is $n$ with probability $\frac{1}{n}$ and $0$ otherwise. The agent's utility $Y_i$ is $M$ with probability $\frac{1}{M}$ independently of $X_i$ and $0$ otherwise. The constraint is a $1$-uniform matroid. We let the agent break ties in favor of the principal.
	
	First, we will determine the principal's optimal non-delegated expected utility. This is given by the solution to Weitzman's Pandora's box problem. For each element $i$, we must determine the cap value $\tau_i$ such that $\E (X_i - \tau_i)^+ = c_i$. It's not hard to verify for this instance that $\tau^x_i = \varepsilon n$. Then the optimal solution guarantees an expected utility of $U = \E \max_i \min(X_i, \tau_i)$ where each $\min(X_i, \tau_i)$ takes value $\varepsilon n$ with probability $\frac{1}{n}$ and $0$ otherwise. Therefore, $\max_i \min(X_i, \tau_i)$ takes value $\varepsilon n$ with probability $1 - \left( 1 - \frac{1}{n} \right)^n$ and the principal gets expected utility
	\begin{equation*}
		\opt = \varepsilon n \left( 1 - \left( 1 - \frac{1}{n} \right)^n \right) \ge \varepsilon n \left( 1 - \frac{1}{e} \right).
	\end{equation*}
	
	Now, we will bound the principal's delegated expected utility. Consider an arbitrary acceptable set $\mathcal R$ that the principal might commit to. Since the constraint is $1$-uniform, $\mathcal R$ consists of a set of singleton outcomes. Observe that every element $i$ evaluates to one of four tagged outcomes $(i, n, M)$, $(i, n, 0)$, $(i, 0, M)$, and $(i, 0, 0)$ with probabilities $\frac{1}{n M}$, $\frac{1}{n} \left( 1 - \frac{1}{M} \right) $, $\frac{1}{M} \left( 1 - \frac{1}{n} \right)$, and $\left( 1 - \frac{1}{n} \right) \left( 1 - \frac{1}{M} \right)$, respectively.
	
	Given $\mathcal R$, let $E^* \subseteq E$ be the subset of elements $i$ for which $(i, 0, M) \in \mathcal R$ and $(i, n, M) \in \mathcal R$, and let $k = |{E^*}|$. Consider any element $i \notin E^*$. If outcome $(i, 0, M) \notin \mathcal R$, then the agent's increase in expected utility from probing $i$ is at most $M \cdot \frac{1}{M} \left( 1 - \frac{1}{n} \right) - (1 - \varepsilon) = \varepsilon - \frac{1}{n} < 0$, so they have no incentive to ever probe $i$. Similarly, if outcome $(i, n, M) \notin \mathcal R$, then the agent's increase in expected utility from probing $i$ is at most $M \cdot \frac{1}{n M} - (1 - \varepsilon) = \varepsilon - \left( 1 - \frac{1}{n} \right) < 0$. Therefore, the agent will probe no more than the $k$ elements in $E^*$. If $k = 0$, then the agent will not probe anything and both will get $0$ utility. For the remainder of the proof, we assume $k > 0$.
	
	The agent now faces an instance of the Pandora's box problem, so their optimal strategy is to probe elements in order of weakly decreasing cap value (among non-negative cap values) and accept the first outcome whose value is above its cap. For all elements $i \in E^*$, we can calculate that the agent's cap is $\varepsilon M > 0$. Then their optimal strategy is to probe elements from $E^*$ in some order $1, \dots, k$ until a value of $M$ appears, which they will propose. If no value of $M$ appears after probing all of $E^*$, then they will stop probing and choose some outcome to propose. Since all probed outcomes have $0$ utility to the agent, they will choose an outcome to propose that maximizes the principal's utility.
	
	Consider the utility that the principal gets when the agent finds an outcome of value $M$. Among the $k = |{E^*}|$ elements that the agent might probe, they find a value of $M$ with probability
		$1 - \left( 1 - \frac{1}{M} \right)^k \le \frac{k}{M} \le \frac{\varepsilon k}{n} \le \varepsilon.$
	
	Since the principal's utility for the proposed outcome is independent of the agent's, it will have value $n$ for the principal with probability $\frac{1}{n}$. Since $k \ge 1$, the principal pays a cost of $1 - \varepsilon$ for the first probe. Therefore, the principal expects a utility of at most $\varepsilon (\frac{n}{n} - (1 - \varepsilon)) = \varepsilon^2$ in the event when the agent finds an outcome with value $M$.
	
	Now, with probability $\left( 1 - \frac{1}{M} \right)^k \ge 1 - \varepsilon$, the agent doesn't find any outcomes of value $M$. Then the principal pays a cost of $k(1 - \varepsilon)$ in order to probe all $k$ elements in $E^*$. Since the agent breaks ties in favor of the principal, they will propose any acceptable outcomes of value $n$ to the principal. There exists such an outcome with probability at most $1 - \left( 1 - \frac{1}{n} \right)^k$. Therefore, the principal expects a utility of at most
	\begin{equation*}
		n \left(1 - \left( 1 - \frac{1}{n} \right)^k \right) - k(1 - \varepsilon) \le n \left(1 - \left( 1 - \frac{1}{n} \right)^k \right) - k \left( 1 - \frac{1}{2n} \right)
	\end{equation*}
	in the event when the agent does not find an outcome with value $M$. At $k = 1$, this expression evaluates to $\frac{1}{2n} = \varepsilon$. At $k = 2$ it evaluates to $0$. With some calculus and some algebraic manipulations, we can show that this expression is negative for all $k > 2$.
	
	Putting everything together, the principal's delegated expected utility is at most $\varepsilon + \varepsilon^2$, while their non-delegated expected utility is at least $\varepsilon n \left( 1 - \frac{1}{e} \right)$. Therefore, the delegation gap on this instance approaches $\frac{1}{n (1 - 1 / e)} = O(\frac{1}{n})$ as $\varepsilon$ approaches $0$.
 \end{proof}

%% file: free-agent-discounted-cost-models.tex
\section{Free-Agent Model}
\label{free-agent-discounted-cost-models}

The impossibility of constant factor delegation for the standard model, as discussed in Proposition \ref{prop:impos1}, motivates us to design efficient delegation strategies for variants of this model as defined in Section \ref{model-variants}. We observe that the impossibility is aided by the fact that the agent's expected utility for each element is very close to the probing cost, so the principal cannot restrict the agent on any element they want to be probed. An initial attempt to circumvent this failure might design a model where the principal can take on a larger proportion of the probing cost so that they can more freely restrict the agent's behavior. However, the principal's expected utility for each element is similarly close to their probing cost, so they cannot take on a large enough share of the cost without their own expected utility becoming negative.

As a new approach to achieving constant delegation gaps, we will now consider delegation in the free-agent model. Recall that this model removes the agent's probing costs but requires that they always break ties in favor of the principal. This model can be applied in settings where it is standard for the principal to incur the total probing cost. As a simple example, an organization (modeled by the principal) might pay the full travel and lodging expenses associated with interviewing candidates for an available position. The interviewer (agent) can then freely choose to interview (probe) candidates and make recommendations of their own choosing.

We will start by showing that there are constant discounted-cost approximations for this model for any constant discount factor $\delta$ and certain downward-closed constraints.

\subsection{Efficient Delegation for the Free-Agent Model with Discounts}
\label{combined-free-agent-discounted-cost-model}

In Proposition \ref{prop:unif_matroids_freagent}, we propose a $(\delta, \delta')$-factor strategy for $k$-uniform matroid constraints for any $0 \leq \delta \leq 1/2$ and $\delta' \geq \delta$. We show that it is possible to design $\delta$-factor agent-agnostic delegation for the free-agent model with a constant discount factor $\delta' \geq \delta$ on costs for $k$-uniform matroid constraints. Recall that $Z_i^{\min} = \min\{X_i,\tau_i\}$, where $\tau_i$ is the solution to $\mathbb E[(X_i-\tau_i)]=c_i$.  

\begin{proposition} \label{prop:unif_matroids_freagent}
 Let $I$ be an instance of the free-agent model with a $k$-uniform matroid constraint. Then there exists a $(\delta,\delta')$-factor delegation strategy for any $0 \leq \delta \leq 1/2$ and $\delta'\geq \delta$.
\end{proposition}
\begin{proof}
	For $0\leq \delta < 1/2$, it is sufficient to prove the theorem for $\delta = \delta'$ as $(\delta,\delta)$-factor delegation is also a $(\delta,\delta')$ delegation for any $\delta'\geq \delta$. Consider the delegation strategy in which the principal sets a threshold $T$ such that $\Pr[|\{i : Z_i^{\min} \geq T\}| \geq k] = \delta$ and restricts the agent to the set of elements $S = \{i: \tau_i \geq T\}$. Among the elements in $S$, they will accept any combination of outcomes of utility at least $T$ (subject to the $k$-uniform matroid constraint):
	\begin{equation*}
		\mathcal R = \{ \{ (i, x_i, y_i) : i \in S_k \} : S_k \subseteq S \text{~and~} |S_k| \le k \text{~and all~} (x_i, y_i) \in \supp(\mu_i) \text{~and all~} x_i \ge T \}
	\end{equation*}

	We will show that $\mathcal R$ achieves an $\delta$-factor of $\opt$ when the principal pays $1 - \delta$ factor of the total probing cost. Now, let's first bound $\opt$:
	
	\begin{align*}
		\opt
		&= \E \left[ \max_{Q : |Q| \leq k} \sum_{i \in Q} Z^{\min}_{i} \right]\\
		&\leq kT + \E \left[ \max_{Q : |Q| \leq k} \sum_{i \in Q} (Z^{\min}_{i}-T)_+\right] \\
		&\leq kT + \sum_{i=1}^n \E[(Z^{\min}_i-T)_+]\\
		&= kT + \sum_{i\in S} \E[(Z^{\min}_i-T)_+]
	\end{align*}
	The last equality holds because for all $i \notin S$, $\tau_i < T$ implies that $Z^{\min}_i < T$. Hence $(Z^{\min}_i - T)_+ = 0$ with probability $1$. Now, we claim that for all $i \in S$, we have $(Z^{\min}_i - T)_+ = (X_i - T)_+ - (X_i - \tau_i)_+$ with probability $1$. Recall that $\tau_i \geq T$ for all $i \in S$. So when $X_i \geq \tau_i \geq T$ we have that $i \in S$: $ (X_i - T)_+ - (X_i - \tau_i)_+ = \tau_i - T = (Z^{\min}_i - T)_+$, and when $X_i < \tau_i$ we similarly get $ (X_i - T)_+ - (X_i - \tau_i)_+ = (X_i - T)_+ = (Z^{\min}_i - T)_+$. Therefore, we can modify the upper bound on $\opt$ as follows:
	\begin{align}
		\opt &\leq kT + \sum_{i\in S}\left\{(X_i-T)_+ - (X_i-\tau_i)_+\right\} \notag \\
		&\leq kT + \sum_{i\in S} \E[(X_i-T)_+] - c(S)
	\end{align}
	
	Now we will lower bound the principal's delegated utility under strategy $\mathcal R$. Since the agent does not pay any probing costs, they will (in the worst case) probe all elements in $S$ and propose a set of elements $E'$ with $X_i \geq T$ for each $i \in E'$ (if such elements exist) that maximizes their value $\sum_{i\in E'}Y_i$. Recall that we assume the agent will not probe any elements for which they have $0$ expected utility and do not benefit the principal, so the agent won't probe any elements outside of $S$.
	
	Let $A$ be the set of elements with $Z^{\min}_i \geq T$. By definition of the threshold T,
	\begin{align*}
		\Pr [|A| \geq k]
		&= \Pr[\exists A \subseteq [n], |A| \geq k, Z^{\min}_i \geq T \text{~for all~} i \in A] \\
		&= \Pr[\exists A \subseteq S, |A| \geq k, Z^{\min}_i \geq T \text{~for all~} i \in A] \\
		&= \Pr[\exists A \subseteq S, |A| \geq k, X_i \geq T \text{~for all~} i \in A] \\
		&= \delta
	\end{align*}
	
	The above equality shows that there will be at least $k$ element in $S$ with $X_i \geq T$ with probability $\delta$, therefore the principal will at least obtain value $kT$ plus some extra value with probability $\delta$. We assume the worst-case behavior from the agent: they probe all elements in $S$, and if $A$ is the set of elements $i$ for which $X_i \geq T$, then the agent proposes a maximal set of elements in $A$ with the minimum $x_i$ values. 
	
	Consider the following three events: $|A| > k$, $1 \leq |A| \leq k$, and $A = \emptyset$. Note that $\Pr[|A| > k ] + \Pr[|A| = k]  =\delta$ and $ \Pr[|A| < k] = 1 - \delta$. Moreover, whenever $|A|\leq k$, the agent will select the entirety of $A$ and propose to the principal because they have no incentive to drop any element $i$ with $x_i\geq T$. 
	
	Now, we can lower-bound the principal's delegated expected utility for the worst-case agent with $1 - \delta$ discount factor as follows:
	
	\begin{align}
		&\E[\del]\notag \\
		&\geq \E[\del ~|~ |A|>k] \cdot \Pr[|A|>k] + \E[\del ~|~ 1 \le |A| \le k] \cdot \Pr[1 \le |A| \le k] + \E[\del|A=\emptyset]\Pr[A=\emptyset] \notag \\
		&\geq kT(\Pr[|A|>k]) + \E[\del | 1 \le |A| \le k] \cdot \Pr[1 \le |A| \le k] - (1 - \delta) c(S) (\Pr[|A|>k]+\Pr[A=\emptyset]) \label{eq:high_value} \\
		&\geq \begin{aligned}[t]
			&k T(\Pr[|A|>k] + \Pr[|A| = k]) - (1 - \delta)c(S) \\
			&+ \sum_{i\in S} \E[X_i - T | X_i \geq T \land |A| \leq k] \Pr[X_i \geq T] \cdot \Pr[|A \setminus i| \leq k - 1]
		\end{aligned} \label{ineq:del} \\
		&\geq \delta k T +  \sum_{i \in S} \E[(X_i - T)_+] \cdot \Pr[|A \setminus i| \leq k - 1] - (1 - \delta) c(S) \notag \\
		&\geq \delta k T  + \sum_{i \in S} \E[(X_i - T)_+] \cdot \Pr[|A| \leq k] - (1 - \delta) c(S) \notag \\
		&\geq  \delta k T + (1 - \delta) \sum_{i \in S} \E[(X_i - T)_+] - (1 - \delta) c(S) \notag \\
		&= \delta \left(k T + \sum_{i \in S} \E[(X_i - T)_+] - c(S) \right) + (1 - 2 \delta) \left( \sum_{i \in S} \E[(X_i - T)_+] - \sum_{i \in S} \E[(X_i -
		\tau_i)_+] \right) \label{eq:separation_of_cost_to_optimal} \\
		&\geq \delta \opt .\label{eq:result_1-uniform}
	\end{align}

	Inequality~\ref{eq:high_value} holds because the principal will obtain at least utility of $kT$ when $|A|\geq k$. Inequality \eqref{ineq:del} holds because when $1\leq |A|\leq k$, the agent will propose the entire set $A$. The inequality \eqref{eq:result_1-uniform} holds because $\delta \leq 1/2$ and $\tau_i \geq T$ for all $i \in S$ implies that $\E[(X_i-T)_+] - \E[(X_i-\tau_i)_+]\geq 0$. This concludes the proof.
%
\end{proof}


We now extend the constant delegation gap for the free-agent model with constant discounts to general downward-closed constraints. In Theorem \ref{thm:efficient_delegation_from_OCRS}, we show a reduction from the free-agent model with constant discounts to selectable greedy OCRS. We show that if there exist $\alpha$-selectable greedy OCRS for the polytope $P_\mathcal I = \operatorname{conv} \{\mathrm{1}_S:S\in \mathcal I\}$ then the principal can construct a $(\alpha,1-\alpha)$-factor delegation strategy for the free-agent model with constraint $\mathcal I$. Theorem~\ref{thm:efficient_delegation_from_OCRS} further implies constant factor delegation for the free-agent model with constant discounts for general matroids, matchings, and knapsack constraints.

\begin{theorem}\label{thm:efficient_delegation_from_OCRS}
	Given an instance of the free-agent model with constraint $\mathcal I$, if there exists an $\alpha$-selectable greedy OCRS for the polytope $P_\mathcal I = \operatorname{conv} \{\mathrm{1}_S:S\in \mathcal I\}$, then there exists a $(\alpha, \delta)$-factor strategy for the given instance where the discount factor $\delta \geq 1-\alpha$.
\end{theorem}
\begin{proof}
	Given an instance of the delegated Pandora's box problem for the free-agent model with elements $E$ and constraint $\mathcal I$, let a random optimal set $I^*$ defined as follows: $I^* = \argmax_{S \in \mathcal I} \sum_{i \in S} Z^{\min}_i$. We define $p_i^* = \Pr[i \in I^*]$ and thresholds $t_i$ such that $\Pr[Z^{\min}_i \geq t_i] = p_i^*$. Notice that $I^* \in \mathcal I$ with probability $1$, we have that $p^*$
	is a convex combination of characteristic vectors of
	feasible sets $\mathcal I$, and hence, $p^* \in P_\mathcal I$. The principal rejects all elements not in $E' =\{i\in E: p_i^*>0\}$, so the agent has no incentive to probe them. Note that for all elements $i \in E'$, $0 < t_i < \tau_i$. We can bound the optimal utility as follows:
	
	\begin{align*}
		\opt
		&\leq \E \left[ \max_{S\in \mathcal I} \sum_{i \in S} Z^{\min}_i \right] = \sum_{i \in E'} \E[Z^{\min}_i ~|~ i \in I^*] \Pr[i \in I^*] \\
		&\le \sum_{i \in E'} \E[Z^{\min}_i ~|~ Z^{\min}_i \geq t_i] p^*_i\\
		&= \sum_{i \in E'} \E[(Z_i^{\min} - t_i)_+] + \sum_{i \in E'} t_i p^*_i \\
		&= \sum_{i \in E'} \E[(X_i - t_i)_+ - (X_i - \tau_i)_+] + \sum_{i \in E'} t_i p^*_i \\
		&= \sum_{i \in E'} \E[(X_i - t_i)_+] + \sum_{i \in E} t_i p^*_i - c(E')
	\end{align*}
	Let $\mathcal I_{p^*} \subseteq \mathcal I$ be the downward closed family generated by $\alpha$-selectable greedy OCRS for $p^* \in  P_\mathcal I$. Now, consider a delegation strategy in which the principal accepts a proposal of elements $S$ if and only if $S \in \mathcal I_{p^*}$ and the realizations of all $i \in S$ is greater than or equal $t_i$, i.e.
	\begin{equation*}
		\mathcal R = \{ \{ (i, x_i, y_i) : i \in Q \} : Q \in \mathcal I_{p^*} \text{~and~} (x_i, y_i) \in \supp(\mu_i) \text{~and~} x_i \ge t_i \text{~for all~} i \in Q \}
	\end{equation*}
	
	Since the agent does not incur any cost for probing, in the worst case they will probe all elements in $E'$. Let $R(t)$ be the set of elements with $X_i \geq t_i$. The agent will always propose some maximal set $I$ with $I \subseteq R(t)$ and $I \in \mathcal I_{p^*}$. More formally, let
	\begin{equation*}
		\mathcal I_{p^*}^{R(t)} = \{S : (S \in \mathcal I_{p^*}) \text{~and~} (X_i \geq t_i \text{~for all~} i \in S) \text{~and~} (S \cup i' \notin \mathcal I_{p^*} \text{~for all~} i' \in E' \setminus S \text{ with } X_{i'} \geq t_{i'})\}
	\end{equation*}
	be the family of sets of elements that the agent might propose. In the worst-case, they will propose some such set of elements that minimizes the principal's utility. We can think of this worst-case agent as follows: an almighty adversary who presents elements in the worst possible sequence for all realizations, the agent then picks element $i$ if and only if $X_i \geq t_i$ and the selected set satisfies the feasible constraints $\mathcal I_{p^*}$ \footnote{An almighty adversary knows the coin flips of the agent's strategy, i.e. $\mathcal I_{p^*}$ and $R(t)$. Therefore, an almighty adversary can force the agent to select any $S\in \mathcal I_{p^*}^{R(t)}$ of their choice.}.
	Note that this agent has no incentive to pick a set outside of $\mathcal I_{p^*}^{R(t)}$. Since $\mathcal I_{p^*}$ is generated by an $\alpha$-selectable greedy OCRS, given any currently selected set by the agent $S \subseteq R(t)$ and $S \in \mathcal I_{p^*}$, we have $\Pr[S \cup i \in \mathcal I_{p^*}] \geq \alpha$. Therefore for all elements $i \in E$, we have $\Pr[i \in I] \geq \alpha \cdot \Pr[X_i \geq t_i]$. Then
	\begin{align*}
		\E[\del]
		&= \E\left[\sum_{i \in I} X_i \right] - (1 - \delta) c(E')\\
		&\geq \sum_{i \in E'} \E[X_i ~|~ i \in I] \cdot \Pr[i \in I] - (1 - \delta) c(E').
	\end{align*}
	Since the agent selects an element $i$ only if $X_i \geq t_i$, on the adversarial arrival of elements selected by an almighty adversary, $\E[X_i ~|~ i \in I] = \E[X_i ~|~ X_i \geq t_i]$. We can bound the principal's expected delegation with a constant discount $\delta \geq 1 - \alpha$ as follows:
	\begin{align*}
		\E[\del]
		&= \sum_{i \in E'} \E[X_i ~|~ i \in I] \cdot \Pr[i \in I] - (1 - \delta) c(E') \\
		&\geq \alpha  \cdot \sum_{i \in E'} \E[X_i ~|~ X_i \geq t_i] \cdot \Pr[X_i \geq t_i] - \alpha \cdot c(E') \\
		&= \alpha \cdot \left\{ \sum_{i \in E'} \E[(X_i - t_i)_+] + \sum_{i \in E} t_i p^*_i - c(E') \right\} \\
		&\geq  \alpha \cdot \opt
	\end{align*}
	Concluding the proof.
\end{proof}

We note that the argument above reduces deterministic delegation, in which the principal chooses their strategy deterministically, to deterministic greedy OCRS. Perhaps surprisingly, it can also reduce deterministic delegation to randomized greedy OCRS as defined in \cite{feldman2016online}. The reason is that any randomized greedy OCRS is randomization over deterministic OCRS, so the reduction constructs a distribution over delegation mechanisms achieving the desired approximation. However, our model of delegation is a Stackelberg game in which the principal moves first, so their best randomized strategy can be no better than their best deterministic strategy. Therefore, the principal can choose the best deterministic strategy from among the distribution provided by the reduction for the same approximation factor.

Theorem~\ref{thm:efficient_delegation_from_OCRS} combined with efficient $\alpha$-selectable greedy OCRS schemes \cite{feldman2016online} implies the following corollary.

\begin{corollary}
	There exist $(\alpha,\delta)$-factor delegation strategies (agent-agnostic) for the free-agent model with matroid, matching, and knapsack constraints and constant discount factor $\delta$. Specifically, these constants for matroids, matchings, and knapsacks are $\alpha =1/4, \delta \geq 3/4$, $\alpha = 1/2e, \delta \geq 1 -1/2e$ and $\alpha = 3/2 - \sqrt 2,\delta \geq \sqrt 2 - 1/2$, respectively.
\end{corollary}

\subsection{Free-Agent Model Impossibility without Discounts}
\label{free-agent-model-impossibility}

One of the primary motivations for introducing this model comes from the impossibility in Section \ref{standard-model-impossibility} and an attempt to circumvent one of the challenges with achieving a constant delegation gap. Recall from that section, the instance for which $X_i = n$ and $Y_i = n$ independently with probability $1 / n$ each and $0$ otherwise. Now that the agent does not pay to probe, the principal may choose accept only outcome $(i, n, n)$ from element $i$ because the agent's expected utility from probing $i$ is $n \cdot \Pr[X_i=n] \Pr[Y_i=n] = 1/n > 0$. However, since the agent does not pay to probe, they may probe all elements that can be accepted with nonzero probability so long as they could do better by probing such elements. Therefore, the agent might incur too large a probing cost for the principal compared to what the principal would pay on their own. In Proposition \ref{prop:impos2}, we describe a family of instances of the free-agent model for which the delegation gap is $O(1/n^{1/4})$ without any discounts. Proposition~\ref{prop:impos2} shows that it is impossible to obtain a constant factor delegation gap for the free-agent model without any discounts, even when the agent breaks all ties in favor of the principal. Moreover, it holds even when the agent does not probe all possible elements whose outcome is acceptable with nonzero probability.

\begin{restatable}{proposition}{propimposfreeagent}\label{prop:impos2}
	There exists an instance of the free-agent model on $n$ elements with a $1$-uniform matroid constraint such that the delegation gap is $O ({1}/{n^{\frac{1}{4}}})$, even when the agent breaks all ties in favor of the principal.
\end{restatable}
We defer the proof of Proposition~\ref{prop:impos2} to Appendix~\ref{appendix:proof_propo5.5} to prevent interruptions to the flow of the paper.

\subsection{Discounted-Cost Impossibility}
\label{discounted-cost-model-impossibility}

With constant-factor delegation gaps for the free-agent model with discounts and an impossibility for the free-agent model, one might hope that the standard model with constant discounts might admit constant delegation gaps. However, we again have an impossibility. In Proposition~\ref{prop:impo42}, we show that there exists a family of instances of the standard model, parameterized by the number of elements $n$, with a generous discount factor $\delta = 1 - 1/\sqrt n$ for which there does not exist any constant factor delegation strategies. Thus, Proposition~\ref{prop:impo42} shows that there can not exist an $(\alpha, \delta)$-strategy for this problem with constants $\alpha$ and $\delta<1$. See 

\begin{restatable}{proposition}{propimposdiscounts}\label{prop:impo42}
	There exist instances of the discounted-cost model on $n$ elements with discount factor $\delta = 1 - 1 / \sqrt n$ (the agent and the principal both pay $(1-\delta) c_i$ for all elements, i.e. $c_i / \sqrt n$) for which the delegation gap is $O\left(1 / \sqrt n \right)$.
\end{restatable}
\begin{proof}
	For any positive integer $n > 1$ and real $\varepsilon = 1 / n^{\frac{1}{4}}$, let $M$ be a positive integer such that $M = \sqrt n$ and consider the following instance of delegated Pandora's box. We have $n$ identical elements $E = \{1, \dots, n\}$ where each element $i$ has a probing cost $c_i = 1 - \varepsilon$ and random utilities $(X_i, Y_i) \sim \mu_i$. The principal's utility $X_i$ is $n$ with probability $\frac{1}{n}$ and $0$ otherwise. The agent's utility $Y_i$ is $M$ with probability $\frac{1}{M}$ independently of $X_i$ and $0$ otherwise. The constraint is a $1$-uniform matroid and there is no outer constraint. We let the agent break ties in favor of the principal.
	
	First, we will determine the principal's optimal non-delegated expected utility. This is given by the solution tothe generalized Pandora's box problem. For each element $i$, we must determine the cap value $\tau_i$ such that $\E (X_i - \tau_i)^+ = c_i$. It's not hard to verify for this instance that $\tau_i = \varepsilon n$. Then the optimal solution guarantees an expected utility of $U = \E \max_i \min(X_i, \tau_i)$ where each $\min(X_i, \tau_i)$ takes value $\varepsilon n$ with probability $\frac{1}{n}$ and $0$ otherwise. Therefore, $\max_i \min(X_i, \tau_i)$ takes value $\varepsilon n$ with probability $1 - \left( 1 - \frac{1}{n} \right)^n$ and the principal gets expected utility
	\begin{equation*}
		\opt = \varepsilon n \left( 1 - \left( 1 - \frac{1}{n} \right)^n \right) \ge \varepsilon n \left( 1 - \frac{1}{e} \right) = \Theta(n^{3/4}).
	\end{equation*}
	Now, we will bound the principal's delegated expected utility when both the agent and the principal get a discount factor of $\delta>1-1/n^{1/2}$. Consider an arbitrary acceptable set $\mathcal R$ that the principal might commit to. Since the constraint is $1$-uniform, $R$ consists of a set of singleton outcomes. Observe that every element $i$ evaluates to one of four tagged outcomes $(i, n, M)$, $(i, n, 0)$, $(i, 0, M)$, and $(i, 0, 0)$ with probabilities $\frac{1}{n M}$, $\frac{1}{n} \left( 1 - \frac{1}{M} \right) $, $\frac{1}{M} \left( 1 - \frac{1}{n} \right)$, and $\left( 1 - \frac{1}{n} \right) \left( 1 - \frac{1}{M} \right)$, respectively.
	
	Given $R$, let $E^* \subseteq E$ be the subset of elements $i$ for which $(i, 0, M) \in R$, and let $k = |{E^*}|$. Consider any element $i \notin E^*$. If outcome $(i, 0, M) \notin R$, then the agent's increase in expected utility from probing $i$ is at most $M \cdot \frac{1}{nM}  - (1 - \varepsilon)(1-\delta)  = \frac 1 n - \frac 1 {\sqrt n}(1-\varepsilon) < 0$ for large enough $n$, so they have no incentive to ever probe $i$. Therefore, for the rest of the proof, we assume that $k > 0$.
	
	The agent now faces an instance of Pandora's box problem, so their optimal strategy is to probe elements in order of weakly decreasing cap value (among non-negative cap values) and accept the first acceptable outcome whose value is above its cap. Note that the agent will only probe the elements that belong to $E^*$ We divide the elements in $E^*$ into the following disjoint sets:
	\begin{align*}
		E^*_1 &= \{i : \{(i,n,M), (i,0,M), (i,n,0)\} \subseteq \mathcal R\}, \\
		E^*_2 &= \{i : \{(i,n,M), (i,0,M)\} \subseteq \mathcal R\}, \\
		E^*_3 &= \{i : \{(i,0,M), (i,n,0)\} \subseteq \mathcal R\}.
	\end{align*}
	The optimal strategy for the agent is to first probe the elements in $E^*_1$ and then $E^*_2$ and stop once they find an outcome with utility $M$. If there is no such outcome, then they probe elements in $E^*_3$ and stops once they find an outcome $(i, 0, M)$. However, the principal has no incentive to construct $\mathcal R$ such that $E^*_2 \neq \emptyset$ or $E^*_3 \neq \emptyset$. For the sake of contradiction, let $E_2^* \neq \emptyset$, in that case, consider an event when the agent does not observe $i\in E^*$ with feasible outcome with $Y_i = M$, however, observes $i' \in E_2^*$ with $(i', n, 0)$. Conditioned on this event, the principal can strictly benefit by adding $(i', 0, n)$ to $\mathcal R$. In all other cases, the principal's utility is unchanged by adding $(i', n, 0)$. Therefore $E_2^* = \emptyset$. Similarly, we can show that the principal strictly benefits by adding $(i, n, M)$ to $\mathcal R$ for all $i \in E_3^*$. Hence, for the rest of the proof, we assume that $E^* = E^*_1$.
	
	Consider the utility that the principal gets when the agent finds an outcome of utility $M$. Among the $k = |{E^*}|$ elements that the agent might probe, they find a utility of $M$ with probability $1 - \left( 1 - \frac{1}{M} \right)^k$. Since the principal's utility for the proposed outcome is independent of the agent's, it will have utility $n$ for the principal with probability $\frac{1}{n}$. Since $k \ge 1$, the principal pays a cost of $1 - \varepsilon$ for the first probe. Therefore, the principal expects a utility of at most
	\begin{equation*}
		\left\{ 1 - \left( 1 - \frac{1}{M} \right)^k\right\} \cdot \left (\frac{n}{n} - (1 - \varepsilon)(1-\delta)\right) = O(1)
	\end{equation*}
	from this part of the agent's strategy.
	
	Now, with probability $\left( 1 - \frac{1}{M} \right)^k$, the agent doesn't find any outcomes of value $M$. Then the principal pays a cost of $k (1 - \varepsilon)$ in order to probe all $k$ elements in $E^*$. Since the agent breaks ties in favor of the principal, they will propose any acceptable outcomes of value $n$ to the principal. There exists such an outcome with probability at most $1 - \left( 1 - \frac{1}{n} \right)^k$. Therefore, the principal expects a utility of at most
	\begin{align*}
		\left( 1 - \frac{1}{M} \right)^k \cdot \left\{n \left(1 - \left( 1 - \frac{1}{n} \right)^k \right) - k(1 - \varepsilon)(1-\delta)\right\}
		&\le \left( 1 - \frac{1}{M} \right)^k	\cdot \left\{k - k(1 - \varepsilon)(1-\delta)\right\} \\
		&\le k (\varepsilon + \delta) \left( 1 - \frac{1}{M} \right)^k
	\end{align*}
	For the sake of exposition, let $f(k) = k \left( 1 - \frac{1}{\sqrt n} \right)^k$. For $k = o(\sqrt n )$, asymptotically, $f(k) = o(\sqrt n)$ and for $k=\omega (\sqrt n)$, $f(k) = \omega (\sqrt n) \mathrm e^{-\frac{\omega (\sqrt n)}{\sqrt n}} = o(\sqrt n)$. For $k=\Theta (\sqrt n)$, $f(k) = \Theta (\sqrt n)$. Therefore, $\max_k f(k) = O(\sqrt n)$ asymptotically.
	
	The above arguments imply that the principal's optimal expected delegation is bounded by $O((\delta + \varepsilon )\sqrt n)+O(1) = O(n^{1/4})$. Hence the delegation gap for the above instance is $O(1/n^{1/2})$.
	
	Note that the impossibility still holds if the principal samples $\mathcal R$ from any distribution $D$ over the sets of feasible solutions. We can similarly show that the optimal distribution $D^*$ over the feasible sets has positive support on the solutions $\mathcal R\in \Omega_\mathcal I$ for which $E^* = E_1^*$. Therefore, for any sample of feasible set $\mathcal R$ from $D^*$, $\E[\del] = O(1/\sqrt n)\opt$. Thus, $\E[\del] = O(1/\sqrt n)\cdot \opt$.
\end{proof}

%% file: shared-costs-model.tex
\section{Shared-Cost Model}
\label{cost-sharing}

We now consider the \emph{shared-cost model}, where the principal decides how to split each probing cost with the agent. This final model gives the principal more control over probing costs in another attempt to get constant-factor delegation gaps despite our previous impossibility results. Recall that in this setting, the principal starts by choosing how to split each probing cost, so that the agent pays $c'_i \in [0, c_i]$ and the principal pays the remaining cost $c_i - c'_i \in [0, c_i]$. This model is motivated not only by our earlier impossibilities, but also by settings in which the principal has the power to pay chosen percentages of different costs that the agent may incur. For example, an organization (modeled by the principal) might reimburse chosen percentages of travel and lodging expenses associated with interviewing candidates based on the total amount of cost and expected quality of the candidate. The interviewer (agent) can then choose to interview (probe) candidates and make recommendations of their own choosing, but they must pay the remaining cost on their own.

In Theorem \ref{thm:delegation_for_costsharing}, we show that there exist efficient constant-factor strategies for the principal for a certain class of downward-closed constraints. This positive result uses a reduction from greedy selectable OCRS to efficient delegation for the shared-cost model.

\begin{theorem}\label{thm:delegation_for_costsharing}
	If there exists an $\alpha$-selectable greedy OCRS for the polytope $P_{\mathcal I} = \operatorname{conv}\{\mathrm{1}_S : S\in \mathcal I\}$, then there exists an $\alpha/2$-factor delegation strategy for the shared-cost model with inner constraint $\mathcal I$.
\end{theorem}
\begin{proof}
	Let $\{ p_i \}_{i \in E}$ be the solution to the following optimization problem:
	\begin{align*}
		&p = \argmax_{q\in P_{\mathcal I}} \sum_{i\in E} g_i(q_i), \quad \text{ where } \quad g_i(p_i) = p_i \cdot \mathbb E[Z_i^{\min} ~|~ Z_i^{\min} \geq F_i^{-1}(1 - p_i)],
	\end{align*}
	where $F_i(z)$ for $i\in E$ is the the cumulative distribution function of $Z_i^{\min}$, similar to \cite{feldman2016online} \footnote{We can also modify the optimization for discrete $Z_i^{\min}$ as in \cite{feldman2016online}.}. For $i\in E$, we set a threshold $t_i = \min\{\beta: F_i(\beta) \geq 1-p_i\}$. For any $p \in P_{\mathcal I}$, let $\mathcal I_p \subseteq \mathcal I$ be the downward-closed set system generated by an $\alpha$-selectable greedy OCRS with marginal probabilities $p$. The proof of Theorem 1.12 from \cite{feldman2016online} shows that for any online/adversarial item arrival order, the simple strategy that selects element $i$ if and only if $X_i\geq t_i$ and $S\cup i \in \mathcal I_p$ (where $S$ is the set of selected elements before the arrival of $i$) obtains at least $\alpha \cdot \E [ \max_{T\in \mathcal I}\sum_{i\in T} Z^{\min}_i]\geq \alpha \cdot \opt$ in expectation. The above strategy is an $\alpha$-factor greedy monotone strategy for the gambler against almighty adversary which can be described as $\mathcal A_t = \{\{(i,x_i): i\in S \}: S\in \mathcal I_{p} \text{ and } x_i \geq t_i \text{~for all~} i\in S\}.$
	
	Given the independent distributions $\{\mu_i\}_{i\in E}$, the principal first computes $d_i = \E[Y_i ~|~ X_i \geq t_i] \cdot \Pr[X_i \geq t_i]$ for each element $i \in E$. If $d_i \le c_i$ for all elements $i \in E$, then the principal selects the agent's costs as $c'_i = d_i$ for all elements. After the cost division, the principal can define their strategy as follows: they accept elements only from the set $F = \{ i \in E : \tau_i\geq t_i \}$ where $\tau_i$ is the principal's cap value for $X_i$. Note that there does not exist $S \in \mathcal I_p$ that contains an element $j\in S$ not belonging to $F$ because the thresholds were defined for the truncated random variable $Z_i^{\min}$. The principal sets the acceptable outcomes as
	\begin{equation*}
		\mathcal R = \{\{(i, x_i, y_i) : i \in S\} : S \in \mathcal I_p \text{~and all~} (x_i, y_i) \in \supp(\mu_i) \text{~and all~} x_i \geq t_i \}.
	\end{equation*}
	Given this delegation strategy, the agent has an expected utility of $\E[Y_i ~|~ X_i \geq t_i] \cdot \Pr[X_i \geq t_i] - c'_i = 0$ for each element $i$ that they might want to probe. Given any set of probed and selected elements $S$, the agent has expected utility $0$ for probing any additional element $i$ such that $S \cup i \in \mathcal I_p$. Hence, the agent has no incentive to deviate from the principal's $\alpha$-factor threshold picking strategy $(\mathcal A_t)$ (from Definition~\ref{def:threshold_picking}) for any probing order, where $\mathcal A_t$ is an $\alpha$-factor greedy monotone strategy for the prophet inequality with random variables $\{Z_i^{\min}\}$ against the almighty adversary defined earlier in the proof. Specifically, if they have already selected elements $S$ and are considering element $i$, they should probe $i$ if and only if $\tau_i\geq t_i$ (otherwise $Z_i^{\min}$ can not be more than $t_i$) and $S \cup i \in \mathcal I_p$, and they should select $i$ if and only if $X_i \geq t_i$. At any given time with selected elements $S$, the agent's expected utility from probing $i$ with $\tau_i\geq t_i$ and $S \cup i \in \mathcal I_p$ is $0$, so there is no incentive to deviate. Since the principal pays at most $c_i$ for the agent to probe each element $i$, Lemma~\ref{lem:reduction-pandora-to-prophet} implies that the principal obtains at least $\alpha \cdot \opt$ by delegating.
	
	However, the agent's expected utility becomes nonzero for feasible elements when there exists some element $i \in E$ with $d_i > c_i$ because then the principal cannot set $c'_i$ any larger than $c_i$. Hence, the agent doesn't have $0$ expected utility for feasible elements and may not follow the principal's optimal search strategy. In such cases, the fact that the principal does not pay to probe helps us get a similar approximation.
	
	Consider the case $d_i > c_i$ for all $i \in E$. If the principal only accepts elements with $X_i \geq t_i$ then they can safely ask the agent to pay the entire cost, i.e. $c'_i = c_i$. Again, consider the same acceptable set discussed earlier in the proof:
	\begin{equation*}
		\mathcal R = \{\{(i, x_i, y_i) : i \in S\} : S \in \mathcal I_p \text{~and~} S \subseteq F \text{~and all~} (x_i, y_i) \in \supp(\mu_i) \text{~and all~} x_i \geq t_i \}.
	\end{equation*}
	Let $\prob$ and $S$ be the set of elements probed and selected, respectively, by the agent for some fixed realization of all random variables. It is easy to observe that there must be no $i \in \prob \setminus S$ with $X_i \geq t_i$ and $S \cup i \in \mathcal I_p$, otherwise the agent can improve their utility by selecting such an element. Moreover, there is no $i \in F \setminus \prob$ with $S \cup i \in \mathcal I_p$, otherwise, the agent can improve their expected utility, given the realizations of elements in $\prob$, by probing element $i$.
	
	Therefore for any fixed realizations, we can consider the agent that executes $\alpha$-factor greedy monotone strategy $\mathcal A_t$ for $\{Z_i^{\min}\}$ for the following element arrival order: first the elements in $S$, then the elements in $\prob \setminus S$, and finally the elements in $F \setminus \prob$. Strategy $\mathcal A_t$ will select all the elements in $S$, but $\mathcal A_t$ will not select any element in $\prob \setminus S$ because, as we already argued, there is no $i \in \prob \setminus S$ with $X_i \geq t_i$. Moreover, $\mathcal A_t$ will not select any element in $F\setminus \prob$ because there is no $i \in F\setminus \prob$ with $S\cup i \in \mathcal I_p$. Therefore, the agent selects exactly the same elements that the $\alpha$-factor greedy monotone strategy $\mathcal A_t$ for $Z_i^{\min}$ would select for the described element arrival order and any realizations. Since the principal does not pay any cost to probe elements, extra elements probed in $\prob$ set do not affect the principal's utility.  Therefore, the principal obtains at least $\alpha \cdot \E [ \max_{T\in \mathcal I}\sum_{i\in T} Z^{\min}_i]\geq \alpha \cdot \opt$ from delegation because $\mathcal A_t$ obtains at least $\alpha \cdot \E [ \max_{T\in \mathcal I}\sum_{i\in T} Z^{\min}_i]$ against the almighty adversary.
	
	Now, finally we consider the case when there are some elements for which $d_i \le c_i$ and others for which $d_i > c_i$. We define $E_1 = \{i \in E: d_i \le c_i\}$ and $E_2 = \{i \in E: d_i > c_i\}$. The principal can restrict the agent to one of these two sets with with the greater expected $\opt$ when they follow the corresponding strategy described above. It is easy to show that the principal only loses at most a factor of $1/2$ in this case compared to the others:
	\begin{align*}
		\opt
		&= \E\left[ \max_{\substack{S_1 \subseteq E_1, S_2 \subseteq E_2 \\ S_1 \cup S_2 \in \mathcal I}}\left(\sum_{i \in S_1} X_i + \sum_{j \in S_2} X_j \right) \right] \leq \E\left[ \max_{\substack{S_1 \subseteq E_1 \\ S_1 \in \mathcal I}} \sum_{i \in S_1} X_i + \max_{\substack{S_2 \subseteq E_1 \\ S_2 \in \mathcal I}} \sum_{j \in S_2} X_j \right]\\
		&\leq 2 \max \left\{ \E\left[ \max_{\substack{S_1 \subseteq E_1 \\ S_1 \in \mathcal I}} \sum_{i \in S_1} X_i \right], \E\left[ \max_{\substack{S_2 \subseteq E_2 \\ S_2 \in \mathcal I}} \sum_{j \in S_2} X_j \right] \right\}
	\end{align*}
	Combining the above arguments, we conclude that there exists an $\alpha/2$-factor delegation strategy for this instance.
\end{proof}

We note that, similarly to Theorem \ref{thm:efficient_delegation_from_OCRS}, this result can reduce deterministic delegation to randomized greedy OCRS.

The following corollary shows that there exists a constant factor delegation gap for the shared-cost model with matroids, matching constraints, and knapsack constraints.

\begin{corollary}
	There exists $\alpha$-factor delegation strategies for matroids, matching constraints, and knapsack constraints for the shared-cost model. Moreover these constants are $\alpha = 1/8$,  $\alpha = 1/4e$ and $\alpha = 3/4 - 1/\sqrt 2$ for the respective constraints.
\end{corollary}
As we discussed in Section~\ref{sec:model_variants}, the delegation gap for instances of the shared-cost model can be greater than $1$, meaning that the principal benefits from delegating (in expectation) and may choose to do so even if they have the ability to conduct the search on their own. However, we can construct an instance of this model for which the delegation gap is strictly less than $1$, showing that this is not possible in general.
\begin{proposition}
	There exists instances of Pandora's box for the shared-cost model with delegation gap $1/2+\varepsilon$ for arbitrary small $\varepsilon > 0$.
\end{proposition}
\begin{proof}
	We can construct an instance with $1$-unifrom matroid constraints similar to \citep[Proposition-4.2]{bechtel2020delegated}. Note that the referenced impossibility has cost $0$ and still holds in the context of the shared-cost model, but we reproduce it here with positive (though negligible) costs.
	
	For small $\varepsilon << 1$, let $X_1 = 1/\varepsilon$ with probability $\varepsilon$ and $0$ otherwise, and $Y_1 = 1 - \varepsilon$ with probability $\varepsilon$ and $0$ otherwise, independently of $X_1$. Let $X_2 = Y_2 = 1$ deterministically and set costs $c_1 = c_2 = \varepsilon^2$. We can compute $\opt = 2-\varepsilon -2\varepsilon^2 - \varepsilon^3 \geq 2-4\varepsilon$.
	
	Consider any cost division $0 \le c'_1 \le c_1$ and $0 \le c'_2 \le c_2$. If the principal accepts element $2$ then the agent will always probe element $2$ and propose. We can enumerate over all possible delegation strategies and show that $\E[\del] \leq 1$ in all cases. This shows that the delegation gap is $1/(2-4\varepsilon$), concluding the claim.
\end{proof}

We observe that the efficient delegation strategy for the shared-cost model constructed in Theorem~\ref{thm:delegation_for_costsharing} relies on a computation of $c'_i$ that uses information about the joint distribution $\mu_i$. In the following proposition, we show that if the principal has no information about the distribution of $Y_i$, then they can not obtain constant factor delegation for the shared-cost model. This holds because, without any information about $Y_i$, the principal does not have enough information to compute a cost division for which they can guarantee that the agent will probe the element $i$. We formalize our intuition in Proposition~\ref{prop:impos3} that shows that the agent agnostic delegation gap for the shared-cost model is at least $O(1/n^{1/4})$. 
\begin{restatable}{proposition}{propimposnoinfoY}\label{prop:impos3}
	There exists a family of instances of the shared-cost model with delegation gap $O(1/n^{1/4})$ when the principal has no information about $\{Y_i\}$.
\end{restatable}
\begin{proof}
	Consider an instance on elements $E$ with $|E| = n$ and a $1$-uniform matroid constraint over $E$. For each element $i$, let the probing cost be $c_i = c = 1 - 2/ n^{1/4}$ and let the principal's utility be $X_i = \sqrt n$ with probability $1/\sqrt n$ and $X_i = 0$ otherwise. Following Proposition \ref{prop:impos2}, we have that $\opt = \Theta(n^{1/4})$. Now, consider any delegation mechanism for the principal for the shared-cost model. Let $c_i' = c_i$ be the cost division for each element $i$ in this mechanism, and let $\mathcal R$ be the set of acceptable solutions. Since the principal has no knowledge of the distributions of the agent's utilities, $\mathcal R$ can only consider the principal's utilities $\{X_i\}$. Let $E_1 = \{i \in E: c_i > 0\}$ and $E_2 = \{i \in E: c_i = 0\}$ be a disjoint partition of $E$.
	
	Now we will define the agent's utilities. For each element $i \in E_1$, let $Y_i \sim \texttt{Unif}[0,c_i/2]$ when conditioned on $X_i = \sqrt n$, and $Y_i = n^2$ deterministically when conditioned on $X_i = 0$. For all $i \in E_2$, let $Y_i \sim \texttt{Unif}[e^n,3e^n]$ independent of $X_i$. First, we need to ensure that the described delegation instance has incentive for the agent to participate when they pay the entire cost, i.e. $\E[Y_i] > c_i$. For each element $i \in E_1$, we have $\E[Y_i] = \E[Y_i~|~X_i = \sqrt n]\Pr[X_i = \sqrt n] + \E[Y_i~|~X_i = 0]\Pr[X_i = 0] > (1-1/\sqrt n)n^2 > c_i$ and for $i\in E_2$, $\E[Y_i] = 2e^{n} > c_i$. Note that the principal has no information about $\{Y_i\}$.
	
	Now, consider any single proposal delegation $\mathcal R = \{\{(i,x_i)\}:i\in E, x_i\in \{\sqrt n , 0 \}\}$. We divide all elements $E$ into following disjoint sets given $\mathcal R$:
	\begin{align*}
		F_1 = \{i\in E_1: (i,\sqrt n ) \in \mathcal R\land (i,0) \notin \mathcal R\} \quad & \quad
		F_4 = \{i\in E_2: (i,\sqrt n ) \in \mathcal R\land (i,0) \notin \mathcal R\} \\
		F_2 = \{i\in E_1: (i,\sqrt n ) \notin \mathcal R\land (i,0) \in \mathcal R\} \quad & \quad
		F_5 = \{i\in E_2: (i,\sqrt n ) \notin \mathcal R\land (i,0) \in \mathcal R\} \\
		F_3 = \{i\in E_1: (i,\sqrt n ) \in \mathcal R\land (i,0) \in \mathcal R\} \quad & \quad
		F_6 = \{i\in E_2: (i,\sqrt n ) \in \mathcal R\land (i,0) \in \mathcal R\}
	\end{align*}
	The agent will never probe elements in $F_1$ because for $i\in E_1$, $\E[Y_i~|~X_i = n]-c_i <0$. The agent's optimal strategy is to probe elements in $V = F_4\cup F_5\cup F_6$ (with $|V| = k$) and pick any feasible element with high $Y_i$. If they can not find any feasible elements in $V$ then they probe elements in $F_3$ then $F_2$ until they observe $X_i = 0$. If they fail to observe an element with $X_i = 0$ then they propose element $i\in F_3$ with maximum $Y_i$. Given the agent's optimal strategy, we can bound the principal's optimal expected delegated utility as follows:
	\begin{align}
		\E[\del_{\mathcal R}] &\leq \E[\del_{\mathcal R}~|~\text{agent finds a feasible $i\in V$}] \cdot \Pr[\text{agent finds a feasible $i\in V$}] \notag \\
		&\quad + \E[\del_{\mathcal R}~|~\text{agent does not find a feasible $i\in V$}] \cdot \Pr[\text{agent does not find a feasible $i\in V$}]\notag \\
		&\leq \sqrt n (1 - \left(1 - 1/\sqrt n\right)^k) - kc \notag \\
		&\quad + \E[\del_{\mathcal R} ~|~ \exists i\in F_2\cup F_3: X_i = 0] \cdot \Pr[\exists i\in F_2\cup F_3: X_i = 0] \notag  \\
		&\quad + \E[\del_{\mathcal R} ~|~\nexists i\in F_2\cup F_3: X_i = 0] \cdot \Pr[\nexists i\in F_2\cup F_3: X_i = 0] \label{eq:bound_when_find_in_V} \\
		& \leq [\sqrt n (1 - \left(1 - 1/\sqrt n\right)^k) - kc ]+ \E[\del_{\mathcal R} ~|~ \nexists i\in F_2\cup F_3: X_i = 0] \cdot \Pr[\nexists i\in F_2\cup F_3: X_i = 0] \label{eq:remove_negative_delegation}\\
		&= O(1) +  (1/\sqrt n)^{|F_2\cup F_3|} \sqrt n = O(1) \label{eq:when_no_x_is_zero}
	\end{align}
	Inequality \eqref{eq:bound_when_find_in_V} holds because $\Pr[\text{agent finds a feasible~} i \in V]$ is bounded by $1$. We can further bound $\E[\del_{\mathcal R} ~|~ \text{agent finds a feasible~}i\in V]$ by assuming that the agent proposes element $i\in V$ with $X_i = \sqrt n$ as long as it exists. Inequality \eqref{eq:remove_negative_delegation} holds because whenever the agent finds $i\in F_2\cup F_3$ with $X_i = 0$, the principal's expected utility is negative, i.e. $\E[\del_{\mathcal R} ~|~ \exists i\in F_2\cup F_3: X_i = 0]\leq 0$. Inequality \eqref{eq:when_no_x_is_zero} holds because $\sqrt n (1 - \left(1 - 1/\sqrt n\right)^k) - kc = O(1)$ for all $k \leq n$ (Proposition \ref{prop:impos2}) and we ignore the cost paid by the principal in $\E[\del_{\mathcal R} ~|~ \nexists i\in F_2\cup F_3: X_i = 0]$. Hence, $\E[\del_{\mathcal R}] = O(1)$. Concluding the proof.
\end{proof}




%% file: open-questions.tex
\section{Open Questions}
\label{open-questions}

In this work, we explored just some of the many possible models and results related to the delegation of the Pandora's box problem. We leave the following open questions for future work.

\begin{itemize}
	
	\item All of our positive results employ deterministic delegation mechanisms. Can the principal do strictly better in any of these models by using a lottery mechanism instead? Note that in Appendix~\ref{appendix:lottery}, we show impossibilities only for the class of binary lottery mechanisms.
	
	\item Can our results be extended to other families of downward-closed constraint systems or even to broader classes of constraints such as prefix-closed constraints \cite{bradac2019near}?
	
	\item We observe that modeling delegation with a constraint system allows us to describe delegation problems in which solutions may not be independently distributed and probing reveals only part of certain solutions. Therefore, it may be interesting to investigate the delegation gap of problems that relax the independence assumption in ways that cannot be represented by the addition of a constraint system.
	
	\item In Theorem~\ref{thm:efficient_delegation_from_OCRS}, we show that there exists a $(\alpha,\delta)$-factor strategy for the free-agent model with discount $\delta \geq 1-\alpha$ for the constraints $\mathcal I$ if there exists $c$-selectable greedy OCRS scheme for a relaxation of $P_\mathcal I$. However, we do not yet know of any impossibility or constant-factor strategy when $\delta < 1- \alpha$.
	
	\item The shared-cost model is unique among the models in this paper for the possibility of delegation gaps strictly greater than $1$, as explained briefly in Section \ref{model-variants}. This is interesting because such a delegation gap could incentivize the principal to delegate a problem that they have the ability to solve on their own, whereas our other models assume that the principal must delegate. Can we characterize the family of instances of the shared-cost model for which the delegation gap is strictly greater than $1$?
	
	\item For the models with strong impossibility results, can we find nontrivial families of instances with ``friendly'' agents which allow the principal to achieve a constant delegation gap?
\end{itemize}

%% file: lottery-mechanisms.tex
\section{Lottery Mechanisms}\label{appendix:lottery}
\label{lottery-mechanisms}

In this section, we consider a class of delegation mechanisms which we call binary lottery mechanisms. These are a class of randomized mechanisms that generalize the deterministic ones used earlier.

Formally, a \emph{lottery mechanism} consists of a menu $\mathcal{R}$ of distributions over solutions. After the principal has announced $\mathcal{R}$ to the agent, they probe elements as usual. However, rather than proposing a single solution to the problem, the agent proposes one of the distributions $D \in \mathcal{R}$ that the principal announced. Then, the principal samples a solution $S \sim D$ from the proposed distribution. If $S$ is a valid solution (feasible in the inner constraint), then the principal accepts and both players receive their respective utilities for $S$ minus the total probing cost. Otherwise, the principal rejects the invalid solution and both players pay the total probing cost with no gain.

A \emph{binary lottery mechanism} is a special case of lottery mechanism in which each distribution $D \in \mathcal{R}$ has support for at most two solutions: one null (status quo) solution and one valid non-null solution. Such a mechanism can be equivalently represented by a set $\mathcal{R}$ of acceptable solutions and a probability $p_S$ for each solution $S \in \mathcal{R}$. Then, the principal accepts proposal $S \in \mathcal{R}$ from the agent with probability $p_S$ and rejects the proposal otherwise. This second representation is the one that we will use for the rest of this section.

Observe that the argument from Section \ref{other-mechanisms} applies only to deterministic multi-round signaling mechanisms. Therefore, such lottery mechanisms may be strictly more powerful than their deterministic counterparts. However, a similar argument can show that we get no increased power from randomized multi-round signaling mechanisms, so it's sufficient to consider only randomized single-proposal mechanisms (lottery mechanisms as defined above).

Since they have fine-tuned control over ``how much'' of each solution to accept, binary lottery mechanisms may seem to give the principal increased delegation power.
However, we will now show that strong impossibilities exist for such mechanisms in the case of the standard model and the free-agent model, generalizing earlier results about deterministic mechanisms.

\begin{proposition}\label{prop:impos_random_1}
	There exist instances of the standard model of delegated Pandora's box on $n$ elements for which the delegation gap is $O(\frac{1}{\sqrt n})$ for the class of binary lottery mechanisms.
\end{proposition}
\begin{proof}

	For any positive integer $n > 1$ and real $0 < \varepsilon = \frac{1}{\sqrt n}$, and consider the following instance of delegated Pandora's box. We have $n$ identical elements $E = \{1, \dots, n\}$ where each element $i$ has a probing cost $c_i = 1 - \varepsilon$ and random utilities $(X_i, Y_i) \sim \mu_i$. The principal's utility $X_i$ is $n$ with probability $\frac{1}{n}$ and $0$ otherwise. The agent's utility $Y_i$ is $2$ with probability $\frac{1}{2}$ independently of $X_i$ and $0$ otherwise. The inner constraint is a $1$-uniform matroid. We let the agent break ties in favor of the principal. Following the poof of Proposition~\ref{prop:impos1}, we have $\opt \geq \varepsilon n \left(1 - 1/e\right) = (1-1/e)\sqrt n$.
	
	Now, we will bound the principal's delegated expected utility. Consider an arbitrary acceptable set $\mathcal{R}$ that the principal might commit to. Observe that every element $i$ evaluates to one of four tagged outcomes $(i, n, M)$, $(i, n, 0)$, $(i, 0, M)$, and $(i, 0, 0)$ with probabilities $\frac{1}{n M}$, $\frac{1}{n} \left( 1 - \frac{1}{M} \right) $, $\frac{1}{M} \left( 1 - \frac{1}{n} \right)$, and $\left( 1 - \frac{1}{n} \right) \left( 1 - \frac{1}{M} \right)$, respectively. We let $p^i_{xy}$ denote the probability chosen by the principal of accepting outcome $(i,x,y)$.
	
	Given $\mathcal R$, let $E^* \subseteq E$ be the subset of elements $i$ for which $(n-1)p^i_{02}+ p^i_{n2}\geq n(1-\varepsilon)$. If any element $i\notin  E^*$  then the agent's increase in expected utility from probing $i$ is at most $2\cdot \frac {p^i_{n2}}{2n} + 2\cdot  \frac{p^i_{02}}{2} \left( 1 - 1/n \right) - (1 - \varepsilon) < 0$, so they have no incentive to ever probe $i$.  Let $|E^*|=k$,therefore, the agent will probe no more than the $k$ elements in $E^*$. If $k = 0$,  then the agent will not probe anything and both will get $0$ utility. For the remainder of the proof, we assume $k > 0$. Note that the principal has no incentive to set $p^i_{00}>0$ and $p^i_{n0}<1$ for any $i\in E^*$. We can use a similar argument as Proposition~\ref{prop:impos2} to show this formally.
	
	The agent now faces an instance of the Pandora's box problem, so their optimal strategy is to probe elements in order of weakly decreasing cap value (among non-negative cap values) and accept the first outcome whose value is above its cap. Thus, the agent probes elements in the decreasing order of cap values $\tau^y_i = (2p_i - 1 + \varepsilon)/p_i$ where $p_i = \frac{p^i_{n2}}{2n} + \frac{p^i_{02}}{2}(1-1/n)$ until the expected gain from an element exceeds the remaining cap values. It is easy to verify that $2 > \tau_i^y \geq0$ for all $i\in E^*$ by the definition of $E^*$.
	
	First, we assume for all elements $i \in E^*$ that $\tau^y_i>0$. Since the cap value is strictly positive for all $i \in E^*$, the agent will never propose an element with $Y_i=0$ if they find $j\in E^*$ with $Y_j = 2$.  Consider the utility that the principal gets when the agent finds an outcome of value $2$. Among the $k = |{E^*}|$ elements that the agent might probe, they find a value of $2$ with probability $1-(1/2)^k$. Since the principal's utility for the proposed outcome is independent of the agent's, it will have value $n$ for the principal with probability $\frac{1}{n}$. Since $k \ge 1$, the principal pays a cost of $1 - \varepsilon$ for the first probe. Therefore, the principal expects a utility of at most $(1-(1/2)^k)(p^i_{n2}\cdot \frac n n - (1- \varepsilon))\leq (1-(1/2)^k)\varepsilon$ from the event when the agent finds some element $i\in E^*$ with $Y_i = 2$.
	
	Now, with probability $\left( \frac 1 2 \right)^k $, the agent doesn't find any outcomes of value $2$. Then the principal pays a cost of $k(1 - \varepsilon)$ in order to probe all $k$ elements in $E^*$. Since the agent breaks ties in favor of the principal, they will propose any acceptable outcomes of value $n$ to the principal. There exists such an outcome with probability at most $1 - \left( 1 - \frac{1}{n} \right)^k$. Therefore, the principal expects a utility of at most
	\begin{equation*}
		n \left(1 - \left( 1 - \frac{1}{n} \right)^k \right) - k(1 - \varepsilon) \le n \left(1 - \left( 1 - \frac{k}{n} \right) \right) - k \left( 1 - \varepsilon \right) = k\varepsilon
	\end{equation*}
	from this event. Hence, $\E[\del] \leq (1-(1/2)^k)\varepsilon + k\varepsilon (1/2)^k \leq O(1)\varepsilon$ for $k\geq 1$. Therefore the delegation gap is $\mathcal O(1/ n)$ when $\tau_i^y > 0$ for $i\in E^*$.
	
	Now, suppose for all elements in $i\in E^*$ that $\tau_i^y = 0$. This implies that $(n-1)p^i_{02} + p^i_{n2} = n(1-\varepsilon)$. In this case, the agent obtains $0$ utility in expectation by probing any element. Thus, the agent will try to break ties in the principal's favor. Let's say the agent probes a set of elements $S$ with observed outcomes $\mathcal S$ where they break ties in favor of the principal at every step. If there exists an element $i$ such that $(i,\cdot,2) \in \mathcal S$, then the agent will never propose an outcome $(j,\cdot,0)$ from $\mathcal S$ because they can obtain better utility by proposing an element $i$ with outcome $(i,\cdot,2)$.
	
	Let $S_t=\{i_1,\dots,i_t\}$ be the set of elements probed by the agent until now with outcome $\mathcal S_t$. Suppose the agent has observed an element $i_\ell \in S_t$ with $Y_{i_\ell} = 2$. In that case, if the agent further probes an element $i$ among the unprobed elements, then they will propose $i$ if and only if $Y_i = 2$. If the agent probes $i$, then the addition in the principal's expected utility is $n\cdot\frac{p^i_{11}}{n}\cdot \Pr[Y_i=2] - c_i <0$. Therefore, the agent will not probe any further elements. Thus, we can conclude that the agent will stop probing elements as soon as they observe an element $i$ such that $Y_i=2$. Similarly, we can show that the agent will stop probing elements if they observe an element $j$ with outcome $(j,n,0)$ before any element $i$ with realization $Y_i = 2$.
	
	We can bound the probability of the agent observing outcome $(\cdot, n, 0)$ before $(\cdot, \cdot, 2)$ by $\frac 1 n ( 1/2 + (1/2)^2 + \dots ) \leq \frac 2 n$. Let us denote the event when the agent finds an element with outcome $(\cdot, n, 0)$ before $(\cdot, \cdot, 2)$ by $\mathcal E_1$. In the event $\mathcal E_1$, the principal obtains value $n$ and pays to probe at least one element. $\E[\del|\mathcal E_1] \leq (n - 1 + \varepsilon)$. In the event $\mathcal E^c$, the agent observes an element with outcome $(\cdot , \cdot, 2)$ before $(\cdot, n, 0)$. In this event, the agent will propose the first observed element $i$ with $Y_i = 2$. Since the principal's utility for the proposed outcome is independent of the agent's, it will have value $n$ for the principal with probability $\frac{1}{n}$. Since $k \ge 1$, the principal pays a cost of $1 - \varepsilon$ for the first probe. Therefore, $\E[\del| \mathcal E^c] \leq (\frac n n - 1 + \varepsilon ) = \varepsilon$. We can now bound the expected delegated utility as follows:
	\begin{align*}
		\E[\del] &= \E[\del | \mathcal E]\Pr[\mathcal E] + \E[\del | \mathcal E^c]\Pr[\mathcal E^c]\\
		&\leq \left(\frac 2 n\right) (n-1-\varepsilon ) + \varepsilon \leq O(1).
	\end{align*}
	Therefore the delegation gap is $\mathcal O(1/\sqrt n)$ when $\tau_i^y = 0$ for $i\in E^*$.
	
	Now, consider the case when $\tau_i^* \geq 0$ for all $i\in E^*$. In this case, the agent first probes elements with positive cap values, and if they are unable to find an element with $Y_i > \tau_i^y$, then they probe elements in $E^*$ with cap value $0$. Therefore we can bound the expected delegation as $\E[\del] \leq O(1)\varepsilon + O(1) = O(1)$. This shows that the delegation gap is $O(1/\sqrt n)$.
\end{proof}

In the following proposition, we show that there exists an instance of the free-agent model in which the delegation gap for binary lottery mechanisms is $O(1/n^{1/4})$.  The instance in Proposition~\ref{prop:impos_random_2} is exactly the same instance described in Proposition~\ref{prop:impos2}. We show that the optimal binary lottery mechanism for the instance described in Proposition~\ref{prop:impos2} coincides with the optimal deterministic mechanism. Hence, the impossibility result for deterministic delegation holds for the class of binary lottery mechanisms as well.

\begin{proposition}\label{prop:impos_random_2}
	There exists an instance of the free-agent model on $n$ elements with a $1$-uniform matroid inner constraint such that the delegation gap is $O ({1}/{n^{\frac{1}{4}}})$ for binary lottery mechanisms, even when the agent breaks all ties in favor of the principal.
\end{proposition}\label{prop:lower_bound}
\begin{proof}
	Consider an instance of the free-agent model with a $1$-uniform matroid inner constraint, and for each element $i$, let $X_i$ and $Y_i$ be independently distributed as follows:
	\begin{equation*}
		\begin{aligned}[c]
			X_i =
			\begin{cases}
				\frac{1}{\sqrt n} & \text{with prob.  } 1/\sqrt n  \\
				0, & \text{otherwise}
			\end{cases}
		\end{aligned}
		\qquad\qquad
		\begin{aligned}[c]
			Y_i =
			\begin{cases}
				\delta_i, & \text{with prob.  }  \frac 1 2  \\
				\mathrm{e}^{n}, & \text{with prob.  } \frac 1 2
			\end{cases}
		\end{aligned}
	\end{equation*}
	where $\delta_i > 0$ are sufficiently small. We set the cost for probing any element $i$ to $c_i = 1 - \varepsilon$, where $\varepsilon = \frac {2}{n^{1/4}}$. Following Proposition~\ref{prop:impos2} $\opt \geq \Theta (n^{1/4})$. For simplicity, let $p = 1/n^{1/4}$.
	
	Now we will bound the principal's optimal delegated expected utility. Consider the delegation strategy defined by some optimal set of acceptable outcomes $\mathcal R$. We let $p^i_{xy}$ denote the optimal probability chosen by the principal of accepting outcome $(i,x,y)$. For ease of notation, we let $p^i_{11} = p^i_{p^2e^n}$, $p^i_{10} = p^i_{p^2\delta_i}$, $p^i_{01} = p^i_{0e^n}$, and $p^i_{00} = p^i_{0\delta_i}$.
	
	For all $i \in E$, we claim that $(p^i_{11} + p^i_{10})/2 > 1 - \varepsilon$ or $p^i_{11} = p^i_{10} = p^i_{01} = p^i_{00} = 0$. Otherwise, if both conditions are broken, the principal obtains $\sqrt n \cdot  (p^i_{11} + p^i_{10})\cdot \frac{1}{\sqrt n} - c\leq 0$ utility in expectation whenever the agent probes element $i$, contradicting the optimality of the principal's strategy. As a result, both $p_{10}$ and $p_{11}$ have to be at least $1 - 2\varepsilon$. We now define the set of elements $E^* = \{ i : p^i_{11} > 0 \text{~and~} p^i_{10} > 0 \}$.
	
	Given $\mathcal R$, the agent's optimal strategy can be described as follows: probe elements one by one in the decreasing order of $\tau^y_i = \frac{p^i_{11}}{2\sqrt n} + \frac{p^i_{01}}{2}(1 - 1/\sqrt n)$ and propose the first element with $Y_i = e^n$ if $p^i_{X_i e^n} \geq \max_j \{ p^j_{11}, p^j_{01} \}$ for unprobed $j \in E^*$. The agent will not stop before observing an element $i \in E^*$ with $Y_i = e^n$ because they can always obtain at least $\frac{p_{11}e^n}{2\sqrt n} \geq (1-2\varepsilon)\frac{e^n}{2\sqrt n} >\delta_i$ in expectation by probing any element. Since the principal wants to maximize the chance of accepting any element $i$ with $X_i = 1/\sqrt n$, they will set $p^i_{11}=1$ and $p^i_{01} = 0$.
	
	If the agent is unable to find such an element, then they will propose some element $i$ for which $Y_i = \delta_i$ with the maximum $p^i_{X_i\delta_i} \cdot \delta_i$. Given the agent's optimal strategy, the principal wants to maximize the chance of accepting an element $i$ with $X_i = 1/\sqrt n$ whenever agent proposes such an element. Therefore, $p^i_{10} = 1$ and $p^i_{01} = 0$ for all $i \in E^*$. We have now shown that the optimal binary lottery mechanism in this instance is exactly the optimal deterministic mechanism discussed in Proposition~\ref{prop:impos2}. Hence, following the proof of Proposition~\ref{prop:impos2}, we conclude that the delegation gap with binary lottery mechanisms for the free-agent model is $O(1/n^{1/4})$.
\end{proof}

%% file: missing-proofs.tex

\section{Proof of Proposition~\ref{prop:impos2}}\label{appendix:proof_propo5.5}

\propimposfreeagent*
\begin{proof}
	Consider an instance of the free-agent model with a $1$-uniform matroid constraint, and for each element $i$, let $X_i$ and $Y_i$ be independently distributed as follow:
	\begin{equation*}
		\begin{aligned}[c]
			X_i =
			\begin{cases}
				\frac{1}{p^2,} & \text{with prob.  } p^2  \\
				0, & \text{with prob.  } 1-p^2
			\end{cases}
		\end{aligned}
		\qquad\qquad
		\begin{aligned}[c]
			Y_i =
			\begin{cases}
				\delta_i, & \text{with prob.  }  \frac 1 2  \\
				\mathrm{e}^{n^2}, & \text{with prob.  } \frac 1 2
			\end{cases}
		\end{aligned}
	\end{equation*}
	where $p = \frac{1}{n^{1/4}}$ and $\delta_i > 0$ are sufficiently small. We set the cost for probing any element $i$ to $c_i = 1 - \frac p 2$ and also observe that $(1 - (1 - p^2)^n) \rightarrow 1$ as $n \rightarrow \infty$.
	
	Once again, the principal's optimal non-delegated expected utility is given by the solution to Weitzman's Pandora's box problem. For each element $i$, we must determine the cap value $\tau_i$ such that $\E (X_i - \tau_i)^+ = c_i$. It's not hard to verify for this instance that $\tau_i = \frac{1}{2p} = \frac{n^{1/4}}{2}$. Then the optimal solution guarantees an expected utility of $\E[\opt] = \E \max_i \min(X_i, \tau_i)$ where each $\min(X_i, \tau_i)$ takes value $\tau_i$ with probability $p^2$ and $0$ otherwise. Therefore, $\max_i \min(X_i, \tau_i)$ takes value $\tau_i$ with probability $1 - \left( 1 - p^2 \right)^n = 1 - \left( 1 - \frac{1}{n^{1/2}} \right)^n = O(1)$ and the principal gets expected utility
	\begin{equation*}
		\opt= O(1) \tau_i = \Theta(n^{1/4}).
	\end{equation*}
	
	Now we will bound the principal's optimal delegated expected utility. Consider the delegation strategy defined by some set of acceptable outcomes $\mathcal R$. Given $\mathcal R$, the agent's optimal strategy (assuming they break ties in favor of the principal) can be described as follows: probe elements one by one for which $(i, 1/p^2, e^{n^2}) \in \mathcal R$ and propose the first observed element with $(x_i, y_i) = (1/p^2, e^{n^2})$. If they are unable to find such an element, then they probe elements with only $(i, 0, e^{n^2}) \in \mathcal R$ and propose the first element with $y_i = e^{n^2}$. Finally, they will probe all other elements in some order and propose any element with maximum $\delta_i$ among probed feasible elements.
	
	For each element $i$, the principal has no incentive to accept only $0$ utility outcomes, so an optimal strategy cannot have both $(i, 1/p^2, e^{n^2}) \notin \mathcal R$ and $(i, 1/p^2, \delta_i) \notin \mathcal R$, since then they may incentivize the agent to probe element $i$ (incurring a cost on the principal) without getting any utility back. Moreover, the principal has no incentive to accept any $0$ utility outcomes from an element $i$ even if they accept at least one of $(i, 1/p^2, e^{n^2})$ or $(i, 1/p^2, \delta_i)$. To see why, consider any delegation strategy $\mathcal R$ for which there exists an element $i$ with $(i, 0, \cdot) \in \mathcal R$. There is a nonzero probability that the agent observes only element $i$ with $Y_i = e^{n^2}$ and $X_i = 0$. In this event, dropping $(i, 0, \cdot)$ from $\mathcal R$ does not change the principal's expected utility. Since the agent breaks ties in favor of the principal, in all other cases they will propose an element $i$ with positive $X_i$. Hence, the principal's expected utility does not decrease if $(i, 0, \cdot) \notin \mathcal R$.
	
	Finally, if $\mathcal R$ is an optimal delegation strategy, then for any element $i \in E$, we have that $(i, 1/p^2, \delta_i) \in \mathcal R$ implies that $(i, 1/p^2, e^{n^2}) \in \mathcal R$ and $(i, 1/p^2, e^{n^2}) \in \mathcal R$ implies that $(i, 1/p^2, \delta_i) \in \mathcal R$. Suppose, for the sake of contradiction, that there exists an element $i$ with only $(i, 1/p^2, e^{n^2}) \in \mathcal R$. Then the agent will probe element $i$ last after probing other elements $i'$ with $(i', 1/p^2, e^{n^2}) \in \mathcal R$ and $(i', 1/p^2, \delta_i) \in \mathcal R$. Now, consider the event in which the agent probes element $i$ and it is the only element with $X_i > 0$ among all the probed elements. The probability of such an event is nonzero. However, the agent will not be able to propose an element $i$ if $Y_i = \delta_i$, which happens with probability $1/2$, and in this case the principal ends up paying the cost for probing $i$ without obtaining any value. By adding $(i, 1/p^2, \delta_i)$ to $\mathcal R$, the principal can increase their expected utility conditioned on $i$ being the only element with $X_i > 0$. In all other cases, adding $(i, 1/p^2, \delta_i)$ to $\mathcal R$ does not affect their utility. This contradicts the optimality of $\mathcal R$.
	
	For the other case, suppose there exists an element $i$ with only $(i, 1/p^2, \delta_i) \in \mathcal R$. Again, the agent first probes the elements with both $(i', 1/p^2, \delta_i) \in \mathcal R$ and $(i', 1/p^2, e^{n^2}) \in \mathcal R$. Consider the event in which they do not observe any element with $i'$ with $X_{i'} > 0$ among the elements probed so far. Now, let assume that the agent probes $i$ right after that (this is the best possible scenario for the principal as all other available elements $i'$ are such that $(i', 1/p^2, \cdot) \notin \mathcal R$). Now if $X_i > 0$ and $Y_i = e^{n^2}$, then the agent will not be able to propose element $i$ and the principal pays the cost for probing $i$ without obtaining any value. Hence, adding $(i, 1/p^2, e^{n^2})$ strictly improves the principal's expected utility in this event, and in all other events, it does not affect their utility.
	
	Now, without loss of generality, we can consider any optimal delegation strategy for the principal defined by a set of feasible elements $A = \{1, \dots, k\}$ for which the principal will accept exactly $(i, 1/p^2, e^{n^2})$ and $(i, 1/p^2, \delta_i)$. Since the agent does not incur any cost, they can probe all $k$ elements and propose their favorite acceptable element. However, we assumed that the agent breaks ties in favor of the principal, therefore they will probe elements one by one and will stop probing as soon as they find an element $j \in A$ with $(j, 1/p^2, e^{n^2})$.
	If the agent can not find any such element, then they will propose $(e, 1/p^2, \delta_e)$ with the maximum $\delta_e$ among probed elements. Now we can bound the principal's optimal delegated expected utility as:
	\begin{align}
		\E[\del]
		&\leq
		\begin{aligned}[t]
			& \Pr[X_1 = 1/p^2] \cdot \Pr[Y_1 = e^{n^2}] \left( \frac{1}{p^2} - c \right) \\
			&+ \Pr[X_2 = 1/p^2] \cdot \Pr[Y_2 = e^{n^2}] (\Pr[Y_1 = \delta_1] + \Pr[Y_1 = e^{n^2}] \cdot Pr[X_1 = 0]) \left( \frac{1}{p^2} - 2c \right) + \dots \\
			&+ \Pr[X_k = 1/p^2] \cdot \Pr[Y_k = e^{n^2}] \prod_{i=1}^{k-1} (\Pr[Y_i = \delta_i] + \Pr[Y_i = e^{n^2}] \cdot Pr[X_i = 0]) \left( \frac{1}{p^2} - kc \right) \\
			&+ \left[ \frac{1}{p^2}\left (1 - \prod_{i=1}^{k}(\Pr[X_i = 0] + \Pr[X_i = {1}/{p^2}] + \Pr[Y_i = \delta_i] \right) - ck \right]
		\end{aligned} \\
		&\leq
		\begin{aligned}[t]
			& \frac{p^2}{2}\left( \frac{1}{p^2} - c \right) + \frac{p^2}{2} \left( 1-\frac{p^2}{2} \right) \left( \frac{1}{p^2} - 2c \right) + \dots \\
			&+ \frac{p^2}{2} \left( 1 - \frac{p^2}{2} \right)^{k-1} \left( \frac{1}{p^2} - kc  \right) + \left[ (1 - (1 - p^2)^k) \frac{1}{p^2} - ck \right]
		\end{aligned} \notag
	\end{align}
	To reduce the clutter, let $r = \left( 1 - {p^2}/{2} \right)$. From Appendix B of \cite{boodaghians2020pandora}, we have that $(1 - (1 - p^2)^k) \cdot {1}/{p^2} -ck \leq 1/2$. Using this, we can simplify the above bound as:
	\begin{align}
		\E[\del]
		&\leq \frac{1}{2} \left\{1 + r + r^2 + \dots + r^{k-1} \right\} - \frac{cp^2}{2} \left(1 + 2r + \dots + kr^{k-1} \right) + \frac{1}{2} \notag \\
		&= \frac{1}{2} \left(\frac{1 - r^k}{1 - r} \right) - \frac{p^2c}{2} \left\{ \left( \frac{1 - r^k}{(1 - r)^2} \right) - \frac{kr^k}{1 - r} \right\} + \frac{1}{2} \notag \\
		&= \frac{1}{p^2} (1 - r^k) - \frac{2c}{p^2} (1 - r^k) + ck r^k + \frac{1}{2} \notag \\
		&= \left( \frac{1}{p} - \frac{1}{p^2} \right) (1 - r^k) + k r^k + \frac{1}{2} \notag \\
		&\leq \frac{1}{2} + O(ne^{-\sqrt n}) \notag
	\end{align}
	The above bound on the expected delegation holds for any budget $k$ and outer constraint to the agent. This shows that the delegation gap is at least $O(n^{1/4})$.
	
	Note that the impossibility still holds if the principal samples $\mathcal R$ from any distribution $D$ over the sets of feasible solutions. We can similarly show that the optimal distribution $D^*$ over the feasibile sets has positive support on the solutions $\mathcal R\in \Omega_\mathcal I$ which can be expressed as $\mathcal R = \{(i,1/p^2,e^{n^2}),(i,1/p^2,\delta_i):i\in A\}$ for some $A\subseteq E$. We earlier showed that for any such $\mathcal R$, $\E[\del] = O(1/n^{1/4})\cdot \opt$. Thus, $\E[\del] = O(1/n^{1/4})\cdot \opt$.
\end{proof}